\documentclass[11pt,fleqn]{article}

\usepackage{geometry}
\usepackage{a4wide}
\usepackage{amsmath}
\usepackage{amssymb}
\usepackage{amsthm}
\usepackage{dsfont}
\usepackage{bbm}
\usepackage{rotating}
\usepackage{enumerate}
\usepackage{epstopdf}
\usepackage[round]{natbib} 
\shortcites{PoSI13}
\usepackage{setspace}
\usepackage{threeparttable}
\usepackage{multirow}
\usepackage[multiple]{footmisc}
\usepackage[tableposition=top]{caption}
\usepackage{subfigure}
\usepackage{lscape}
\usepackage[breaklinks]{hyperref}
\usepackage[textfont={footnotesize},labelfont={footnotesize}]{caption}
\usepackage[x11names]{xcolor}

\usepackage[title]{appendix}

\hypersetup{pdfborder = 0 0 0,%
    pdftitle = Inference for Impulse Responses under Model Uncertainty,%
    pdfauthor = Lenard Lieb and Stephan Smeekes,%
    setpagesize = false,%
    pdfpagelayout = SinglePage,%
    bookmarksopen = true,%
    bookmarksnumbered = true,%
    pdfstartview = Fit,%
    linkcolor = black,%
    anchorcolor = black,%
    filecolor = black,%
    citecolor = black,%
    menucolor = black,%
    urlcolor = blue,%
    colorlinks = true,%
}

\newtheoremstyle{rem}%
{10pt}%
{10pt}%
{}%
{}%
{\bf}%
{.}%
{.5em}%
{}%

\newtheoremstyle{assm}%
{\topsep}%
{\topsep}%
{\normalfont}%
{}%
{\bfseries}%
{}%
{\newline}%
{}%

\newtheoremstyle{alg}%
{\topsep}%
{\topsep}%
{\normalfont}%
{}%
{\bfseries}%
{}%
{\newline}%
{\thmname{#1}\thmnumber{ #2}:\thmnote{ #3}}%

\theoremstyle{assm}
\newtheorem{assumption}{Assumption}

\theoremstyle{alg}
\newtheorem{algorithm}{Algorithm}
\theoremstyle{plain}
\newtheorem{theorem}{Theorem}

\newtheorem{proposition}{Proposition}
\theoremstyle{rem}
\newtheorem{remark}{Remark}

\DeclareMathOperator{\Prob}{\mathbb{P}}

\DeclareMathOperator*{\argmax}{arg~max}

\DeclareMathOperator{\PoSI}{PoSI}
\DeclareMathOperator{\WIMP}{WIMP}

\title{\sc Inference for Impulse Responses under Model Uncertainty%
\thanks{We thank Marco Avarucci, Nalan Bast\"urk, Hanno Reuvers and Peter Schotman for their very helpful discussions and suggestions. We also thank conference and seminar participants at the CFE 2015, London, the NESG 2016, Leuven, and the econometrics seminar at the University of Cologne for their constructive comments. The second author thanks the Netherlands Organization for Scientific Research (NWO) for financial support.}
}

\author{Lenard Lieb%
\thanks{Department of Economics, Maastricht University, P.O. Box 616, 6200 MD Maastricht, The Netherlands. E-mail: \href{mailto:L.Lieb@maastrichtuniversity.nl}{L.Lieb@maastrichtuniversity.nl}}
\\
\and
Stephan Smeekes%
\thanks{Department of Quantitative Economics, Maastricht University, P.O. Box 616, 6200 MD Maastricht, The Netherlands. E-mail: \href{mailto:S.Smeekes@maastrichtuniversity.nl}{S.Smeekes@maastrichtuniversity.nl}}
}
\date{\today}
\begin{document}
\maketitle

\begin{abstract}
In many macroeconomic applications, confidence intervals for impulse responses are constructed by estimating VAR models in levels - ignoring cointegration rank uncertainty. We investigate the consequences of ignoring this uncertainty. We adapt several methods for handling model uncertainty and highlight their shortcomings. We propose a new method – Weighted-Inference-by-Model-Plausibility (WIMP) - that takes rank uncertainty into account in a data-driven way. In simulations the WIMP outperforms all other methods considered, delivering intervals that are robust to rank uncertainty, yet not overly conservative. We also study potential ramifications of rank uncertainty on applied macroeconomic analysis by re-assessing the effects of fiscal policy-shocks.

\end{abstract}

\onehalfspacing

\section{Introduction}
Vector autoregressions (VAR) and, more importantly, their implied impulse responses (IR) are essential tools for applied macroeconomists to investigate the dynamic propagation of (structural) shocks. While VARs fitted to macroeconomic data can incorporate information about unit roots and possible cointegration relations, this evidence is regularly ignored in applied work and inference for IR coefficients is usually based on the VAR specification in levels or first-differences. A common argument for the specification in levels is that estimation by ordinary least-squares (OLS) and the associated traditional approach to inference -- for example via an asymptotically normal \citep{Luetkepohl90} or a bootstrap \citep{Kilian98} approximation -- `allows' for the presence of cointegration. Indeed the level specification results in consistent estimates of the VAR parameters regardless of the true underlying cointegration relations, and, for a fixed horizon, associated inferential procedures remain valid for inference on IR coefficients. However, albeit asymptotically valid, confidence intervals may have poor coverage in small samples when the data are highly persistent and when considering responses at ``longer'' horizons \citep{KillianChang00}. \citet{Phillips98} shows theoretically that if one (or more) unit roots are present, confidence bands based on the normal approximation become invalid at ``(very) long horizons'', while \citet{InoueKilian02} and \citet{Mikusheva12} show that the bootstrap also becomes invalid at such increasing horizons.

These seemingly contradicting theoretical results depend on the asymptotic framework considered; or more precisely on the notion of ``(very) long horizons''. If the considered horizon is kept fixed while the sample size is growing, one arrives at standard asymptotic results. However, if the horizon is modelled as a constant proportion of the sample size, the asymptotic distribution becomes non-standard if (near) unit root(s) are present. Similarly, inference via an asymtotically normal approximation based on a wrongly specified vector error correction (VECM) formulation of the VAR becomes invalid at long horizons as well \citep{Elliott98}. Also, it is well known in the bootstrap literature that misspecification of the cointegration rank leads to an invalid bootstrap procedure \citep{Choi05,InoueKilian02,Mikusheva12}.

Within this growing horizon framework, \citet{PesaventoRossi06} construct confidence intervals for ``long-horizon'' IRs using local-to-unity asymptotics. The resulting confidence bands differ substantially from those obtained through traditional approaches, and suffer in turn from size distortions in short to medium horizons. Moreover, their proposed approach to inference does not account for the possibility of near cointegration, limiting its usefulness for applied work.  \citet{Mikusheva12} proposes a procedure that works uniformly well over the entire parameter space and the entire trajectory of the IRs, but her approach only allows for the construction of uniformly valid inference if at most one ``uncertain'' (unit) root is present in the VAR. Furthermore, her suggested inferential procedure is computationally very expensive even for bivariate VARs, let alone VARs of dimensions usually considered in applied research. Similar settings and problems are considered by \citet{Gospodinov04,Gospodinov10}, \citet{GMP11}, \citet{InoueKilian19}, \citet{PesaventoRossi07} and \citet{Wright00} among others, but all consider at most one unknown root near unity. This setting does not allow for uncertainty about the number of cointegrating relations (if any), which we face in practice. \citet{GHP13} consider the more general setting in an extensive simulation study and conclude that the applied researcher is best advised to estimate the system in levels and construct inference in a traditional way. \citet{JMP13} propose an averaging approach for impulse responses of potentially cointegrated VAR models, but their approach still requires a pre-selection of rank, and does not deal with inference explicitly.

In this paper we re-assess the construction of bootstrap confidence intervals for IRs in persistent, possibly non-stationary VARs. Our main intention is to provide the applied researcher with a reliable and robust alternative to the traditional ``levels'' approach, independent of the IR horizon of interest. We approach the issue of choosing the cointegration rank from a model selection perspective, and consider (bootstrap) methods initially designed to overcome model selection uncertainty in different contexts. In particular, we adapt the endogenous lag selection procedure of \citet{Kilian98b}, the model averaging estimators of \citet{HjortClaeskens03} and the bagging approach proposed by \citet{Efron14} to the rank selection problem in VECMs. As elaborated by \citet{LeebPoetscher05}, inference after model selection is difficult, and there is no guarantee that the above-mentioned methods can solve the problems in our setting.

Therefore, we draw inspiration from the Post-Selection Inference (PoSI) approach of \citet{PoSI13}, which explicitly deals with inference after model selection, to propose a novel way of constructing confidence bands by combining intervals of models for any rank. In our approach, labeled as \textit{Weighted Inference by Model Plausibility} (WIMP), upper and lower bounds of all associated fixed-rank intervals are combined depending on the relative evidence for, or plausibility of, each model. Unlike many approaches considered in the VAR literature, our method does not require any pre-selection of ranks; that is, no pre-testing or selection using economic theory is needed. Instead, the method is fully agnostic about the cointegration rank and is fully data-driven. We provide some simple theoretical results establishing pointwise asymptotic validity of our method under general conditions. Our WIMP intervals tend to deliver coverage probabilities close to nominal levels across the entire trajectory of the IRs, even for ``difficult'' situations where cointegrating relations are very weak. Simulation-based evidence also suggests that the WIMP intervals generally outperform all other considered methods, including the traditional ``level'' approach to inference.\footnote{An alternative way to account for rank uncertainty is to consider lag-augmentation, where the VAR in levels is estimated with an additional lag. \citet{TodaYamamoto95} and \citet{DoladoLuetkepohl96} show that Wald tests on the VAR parameters remain valid regardless the order or (co)integration if one lag too many (i.e.~$p$+1) is added to the VAR model, and only the first $p$ lags are used for subsequent analyses. \citet{KilianLuetkepohl17} and \citet{InoueKilian19} suggest this approach for inference on impulse responses as well. However, neither its theoretical nor its small sample properties have been properly investigated in the literature for impulse response analysis. Moreover, combining the lag-augmentation with a bootstrap procedure is no trivial task and would require further study. Notwithstanding these shortcomings, we considered the lag-augmentation approach in our simulation study, where it is shown to perform considerably worse than the WIMP method.}

While we focus on frequentist inference in this paper, it is worth mentioning that rank uncertainty could also be tackled in a Bayesian VAR framework. However, in many Bayesian applications, uncertainty regarding the cointegration rank is often not taken into account explicitly. Although conceptually different, the Bayesian approach to cointegration is often similar in nature to the construction of classical (likelihood-based) inference. That is, the posterior distribution of (impulse response) parameters is often derived conditional on a pre-determined rank, selected using the marginal likelihood or other model comparison approaches \citep[see for example][for a recent survey]{DelNegro11}. However, several approaches incorporating uncertainty about the cointegration rank when analyzing VARs have been suggested in the Bayesian literature. For instance, \citet{Villani01}, \citet{StrachanVanDijk07}, \citet{KPS08} and \citet{StrachanVanDijk13} propose a Bayesian model averaging scheme, similar in spirit to the approach discussed in Section \ref{sec:ma} below. Alternatively, some authors have suggested various priors on the cointegration relations obtained using economic theory \citetext{see e.g.~\citealt{DSSW07} or \citealt{GLP16} and references therein}, which is a different conceptual approach than our fully data-driven, agnostic approach. Moreover, an explicit (theoretical) investigation of the (joint) posterior distribution of impulse responses of VARs under uncertainty on the (co-)integration relations is, however, limited also in the Bayesian literature.

Since uncertainty about the true cointegration rank is mostly ignored in applied macroeconomic research, we investigate to what extend our more robust approach(es) may change the interpretation of results in practice. More specifically, we re-evaluate the effects of fiscal policy based on four influential structural VAR frameworks. Considering  \citeauthor{BlanchardPerotti02}'s \citeyearpar{BlanchardPerotti02} recursive identification strategy, \citeauthor{MountfordUhlig09}'s \citeyearpar{MountfordUhlig09} sign-restriction approach based on penalty functions, \citeauthor{Ramey11}'s \citeyearpar{Ramey11} narrative VAR framework, and \citeauthor{MertensRavn14}'s \citeyearpar{MertensRavn13,MertensRavn14} proxy-VAR, we find that neglecting rank uncertainty might lead to misleading results. As a companion to this paper, a ready-to-use MATLAB toolbox for the WIMP approach combined with various SVAR identification schemes is available online.\footnote{\href{http://www.stephansmeekes.nl}{www.stephansmeekes.nl}}

The remainder of this paper is organized as follows. In Section \ref{sec:IR_VAR} we discuss standard (bootstrap) approaches to inference in cointegrated VARs and illustrate empirically potential ramifications of rank misspecification. Section \ref{sec:r_unc} first discusses several approaches considered in the literature about model uncertainty and their adaptations to account for rank uncertainty, and next introduces the WIMP method. The performance of the suggested methods is investigated by simulation in Section \ref{sec:sim}. Fiscal policy under rank uncertainty is analyzed in Section \ref{sec:fis}. Section \ref{sec:conc} concludes. Appendices \ref{sec:LAVAR} and \ref{sec:data} contain additional simulation results and data descriptions, respectively.

\section{Bootstrap Inference for Impulse Responses} \label{sec:IR_VAR}

\subsection{The Cointegrated VAR Model and Impulse Responses}
Consider the $K$-dimensional structural vector autoregressive (SVAR) time series process $y_t = (y_{1,t}, \ldots, y_{K,t})^\prime$ observed at $t=1, \ldots,T$:
\begin{equation} \label{eq:SVAR}
B_0 y_t = \sum_{j=1}^{p} B_j y_{t-j} + \varepsilon_t,
\end{equation}
where $\varepsilon_t$ is a $K$-dimensional vector of contemporaneously and serially uncorrelated, weakly stationary structural shocks and $B_0$ is the invertible contemporaneous impact matrix. Pre-multiplying both sides of $\eqref{eq:SVAR}$ with $B_0^{-1}$, we obtain the reduced-form VAR
\begin{equation} \label{eq:VAR}
y_t = \sum_{j=1}^{p} A_j y_{t-j} + u_t,
\end{equation}
where $A_j = B_0^{-1} B_j$ and $u_t = B_0^{-1} \varepsilon_t$.

Define the lag polynomial $A(z)$ as $A(z) = I_k - \sum_{j=1}^p A_j z^j$, such that we can write $A(L) y_t = u_t$, where $L$ is the lag operator $L^j y_t = y_{t-j}$. We now formulate assumptions that allow $y_t$ to be (co)integrated with $r$ cointegrating relations, which we label the `$I(1,r)$ conditions' as in \citet{CRT12}.\footnote{Note that we do not necessarily require that all elements in $y_t$ are integrated of order one. That is, some series may be $I(0)$. In this case cointegration is of a trivial form, as any linear combination of $I(0)$ series remains $I(0)$.}

\begin{assumption}[$I(1,r)$ conditions] \label{ass:1}
\leavevmode \vspace{-\baselineskip}	
\begin{enumerate}[(i)]
\item $A(z)$ has exactly $K- r$ roots equal to 1 and all other roots are outside the unit circle.
\item Defining $\Pi = A(1)$, we have that $\Pi = \alpha \beta^\prime$ for $K \times r$ matrices $\alpha$ and $\beta$ with full column rank, with the implicit definition that $\alpha \beta^\prime = 0$ when $r=0$.
\end{enumerate}
\end{assumption}

If $y_t$ satisfies the $I(1,r)$ conditions, we can write $y_t$ as a VECM
\begin{equation} \label{eq:VECM}
\Delta y_t = \Pi y_{t-1} + \sum_{j=1}^{p-1} \Gamma_j \Delta y_{t-j} + u_t, \qquad t=1, \ldots,T,
\end{equation}
where $\Gamma_j = -\sum_{i=j+1}^p A_j$ for $j=1,\ldots, p-1$.

We can invert the VAR model \eqref{eq:VAR} to obtain the moving average representation 
$y_t = \sum_{j=0}^{t-1} \Psi_j u_{t-j} = \sum_{j=0}^{t-1} \Psi_j B^{-1}_0 \varepsilon_{t-j}$, %
where the $\Psi_j$ matrices contain the reduced-form (i.e.~forecast error) impulse responses and $\Phi_j = \Psi_j B_0^{-1}$ the structural impulse responses. However, as $B_0$ is not identified, we cannot obtain $\Phi_j$ in a unique way, and estimating the structural shocks and their impulse responses requires imposing a particular identification scheme. For that purpose, let $P$ be a $K \times K$ matrix such that $PP^\prime = \Sigma_u$, where the specific form of $P$ depends on the identification method. Then we define the identified structural impulse responses as $\Phi_j = \Psi_j P$. In Section \ref{sec:fis} we discuss several ways to identify the structural shocks.\footnote{As the impulse responses only depend on the cointegration parameters $\beta$ through their product with the loadings $\alpha$, that is through the error correction term $\Pi = \alpha \beta^\prime$, we are not concerned with identification of $\beta$, unlike the setting where inference on the long run relations themselves is the objective.}

For ease of notation later on, we directly link the impulse responses to the VECM parameters. Let $\theta = vec(\Pi, \Gamma_1, \ldots, \Gamma_{p-1})$ denote the vector of VECM parameters. Then we can define $\Psi_j = f_j (\theta)$ and $\Phi_j = f_j (\theta) P$ for $j = 0, \ldots, t-1$, where the nonlinear functions $f_j(\cdot)$ are defined implicitly through inverting the VAR model.

\subsection{Inference Conditional on a Selected Rank}
We can estimate the VECM \eqref{eq:VECM} for a given rank $r$ using the Gaussian quasi maximum likelihood estimator of \citet{Johansen95} to obtain estimates $\hat{\theta}^{(r)} = (\hat{\Pi}^{(r)}, \hat{\Gamma}_{1}^{(r)}, \ldots, \hat{\Gamma}_{p}^{(r)}, \hat{\Sigma}_u^{(r)} )^\prime$, where the superscript $(r)$ emphasizes that estimation is conditional on $r$. Note that $\hat{\Pi}^{(r)} = \hat{\alpha}^{(r)} \hat{\beta}^{(r)\prime}$ and $\hat{P}^{(r)}$ is an estimate of $P$ such that $\hat{P}^{(r)} \hat{P}^{(r)\prime} = \hat{\Sigma}_u^{(r)}$, with $\hat{\Sigma}_u^{(r)}$ the residual variance estimator from the VECM. From inverting the VAR representation of the model, we can then straightforwardly obtain the estimates of the moving average terms, $\hat{\Psi}_0^{(r)}, \ldots, \hat{\Psi}_h^{(r)}$, where $h$ is the (maximum) horizon we are interested in. Specifically, we define the estimated impulse responses as $\hat{\Psi}_{j}^{(r)} = f_j (\hat{\theta}^{(r)})$ and $\hat{\Phi}_{j}^{(r)} = f_j (\hat{\theta}^{(r)}) \hat{P}^{(r)}$, for $j=0,\ldots,h$. 

To account for deterministic components, we can first regress $y_t$ on a constant and possibly a linear time trend to obtain the detrended series $\tilde{y}_t = y_t - \hat{\mu}_0 - \hat{\mu}_1 t$ for $t=1,\ldots,T$ and estimate the VECM without deterministic components on $\tilde{y}_t$ (see also Remark \ref{rem:detr}).

Now consider a general impulse response $\zeta$, which is the object of interest of the analysis. Typically, this would be an element of either $\Psi_j$ or $\Phi_j$ for a certain $j$; that is, $\zeta = \psi_{j,a,b}$ or $\zeta = \phi_{j,a,b}$, where the subscript `$a,b$' indicates the $(a,b)$-th element of the matrix. It might also be a combination of elements; for example, if one wants to perform simultaneous inference across horizons, using the ideas proposed in \citet{BruderWolf17} and \citet[Section 3.6]{LSW15}, we could take $\zeta = \max_{0 \leq j \leq h} \psi_{j,a,b}$, $\zeta = \max_{0 \leq j \leq h} \phi_{j,a,b}$, or its studentized versions. Similarly, one could take the Wald statistics of \citet{InoueKilian16} as $\zeta$. The bootstrap algorithm works the same regardless of the specific object of interest; writing $\zeta$ for a general object of interest simply avoids too cumbersome notation and the need to be specific about its particular form. Regardless of the specific form of $\zeta$, it will be a function of the VAR model parameters $\theta$, and its estimator $\hat{\zeta}^{(r)}$ will be the same function of the VAR parameter estimators $\hat{\theta}^{(r)}$, that is, $\zeta = \bar{f} (\theta)$ and $\hat{\zeta}^{(r)} = \bar{f} (\hat{\theta}^{(r)})$, where the form of the function $\bar{f} (\cdot)$ depends on the desired object of interest.

Various algorithms can be used to construct bootstrap confidence intervals for $\zeta$. In the simulation and empirical sections we use straightforward algorithm based on \citeauthor{Hall92}'s (\citeyear{Hall92}) bootstrap percentile interval, which has regularly been considered in the literature, see e.g. \citet{BLW01}. Details are provided in Appendix \ref{sec:alg}. Other common bootstrap methods that are used include \citeauthor{Efron79}'s (\citeyear{Efron79}) percentile interval and \citeauthor{Kilian98}'s (\citeyear{Kilian98}) bias-corrected bootstrap. Irrespective of the specific choices that can be made, all these algorithms have in common that they generate a bootstrap sample, say $\{y_t^*\}_{t=1}^T$, that has a fixed cointegrating rank $r$. Bootstrap impulse responses are then estimated from this bootstrap sample and used to set up a confidence interval of the form $[L^{(r)}(\gamma), U^{(r)}(\gamma)]$, where the superscript `$(r)$' again highlights the dependence on the chosen rank $r$, and $\gamma$ is the desired confidence level. Hence, the bootstrap adds a second layer of potential rank misspecification next to the estimators themselves, which turns out to lead to further complications if one wants to account for rank uncertainty, as we discuss in Section \ref{sec:r_unc} below. Before discussing methods that potentially can account for rank uncertainty, we illustrate the perils of rank misspecification next.

\subsection{Effects of Rank Misspecification} \label{sec:rankmis}

Standard bootstrap inference assumes knowledge of the true cointegrating rank, labeled as $r_0$; if $r \neq r_0$, inference on $\zeta$ will be inappropriate, in particular for longer horizons. If the chosen rank $r$ is smaller than the true rank, the estimated IRs converge to `pseudo-true' values $\theta_j^{(r)}$ which are different from the true ones. This arises because the VAR parameters converge to their pseudo-true values which satisfy the (incorrect) rank restriction, c.f.~\citet{CRT12}. While in this case bootstrap inference remains valid for the pseudo-true parameters, these parameters can be substantially different from the true IRs, making their interpretation and therefore inference somewhat meaningless, in particular as one typically tries to uncover structural effects which requires knowledge of true parameters.

On the other hand, if $r > r_0$, as for instance in the VAR in levels specification, the short (fixed $j$) and medium ($j/n \rightarrow 0$) horizon IRs are estimated consistently, but at long horizons ($j \sim n$) IRs are inconsistent and even random and inference becomes invalid \citep{Phillips98}.\footnote{Consistency of the estimated IRs also depends on the type of identification considered. For example, under long-run identification the short-run IRs are also not estimated consistently, see e.g.~\citet{Gospodinov10}.} The inconsistency is caused by the domination of the error correction terms for the long-horizon IRs, and their insufficient estimation accuracy under rank misspecification. The same occurs for bootstrap inference; while valid for short and medium horizon IRs, it becomes invalid at long horizons, as demonstrated in different contexts by \citet{Choi05}, \citet{InoueKilian02} and \citet{Mikusheva12}.

Figure \ref{Example_Intro} illustrates potential consequences of rank uncertainty for the construction of inference in practice. Displayed in the left panel are confidence intervals for output responses to a government spending shock identified as in \citet{BlanchardPerotti02} for all possible numbers of cointegration relations.\footnote{The VAR specification and the data are described in Section \ref{sec:fis}.} Clearly, the assessment of the effectiveness of the spending policy varies drastically with the chosen cointegration rank, indicating that choosing the wrong rank hampers the interpretation of results -- for long but equally so for short horizons. One could argue that with proper rank estimation, the most appropriate of these intervals can be selected. However, as demonstrated in the right panel, if evidence for a particular rank is weak, different but equally well established ``respectable'' rank selection procedures may suggest different models, providing little guidance for the applied researcher.

Finally, note that the unrestricted VAR in levels gives substantially different (and narrower) intervals than the VAR models with reduced rank, even the model with the next highest rank ($r=9$). Of course, if the true model is indeed a VAR of full rank, all variables are stationary and no (co)integration would be present. However, many macroeconomic series exhibit persistent behavior, which may be caused by stochastic trends. Indeed, ADF tests cannot reject a unit root for most series in our dataset, casting doubt on whether the levels specification is indeed the most appropriate one. If the series are really cointegrated, a reduced-rank VAR model would be more appropriate and constructing inference based on the VAR in levels would be invalid for long horizons. In practice, distinguishing long from short (or medium) horizons is difficult, and as we show in the simulations, for sample sizes compared to this particular example, inference based on the VAR in levels becomes inaccurate at fairly short horizons already.

As Figure 1 shows, the imposed rank matters for the interpretation of the results, and a ``robust'' decision to use the VAR in levels could, in this example, lead to a misguided interpretation of the IRs. The strategy to use the VAR in levels based on a robustness argument therefore appears questionable, while rank selection techniques also do not appear to give conclusive answers. It is therefore crucial to take rank uncertainty into account when conducting inference for impulse responses.

\begin{figure}
\begin{center}
\includegraphics[width=1\linewidth]{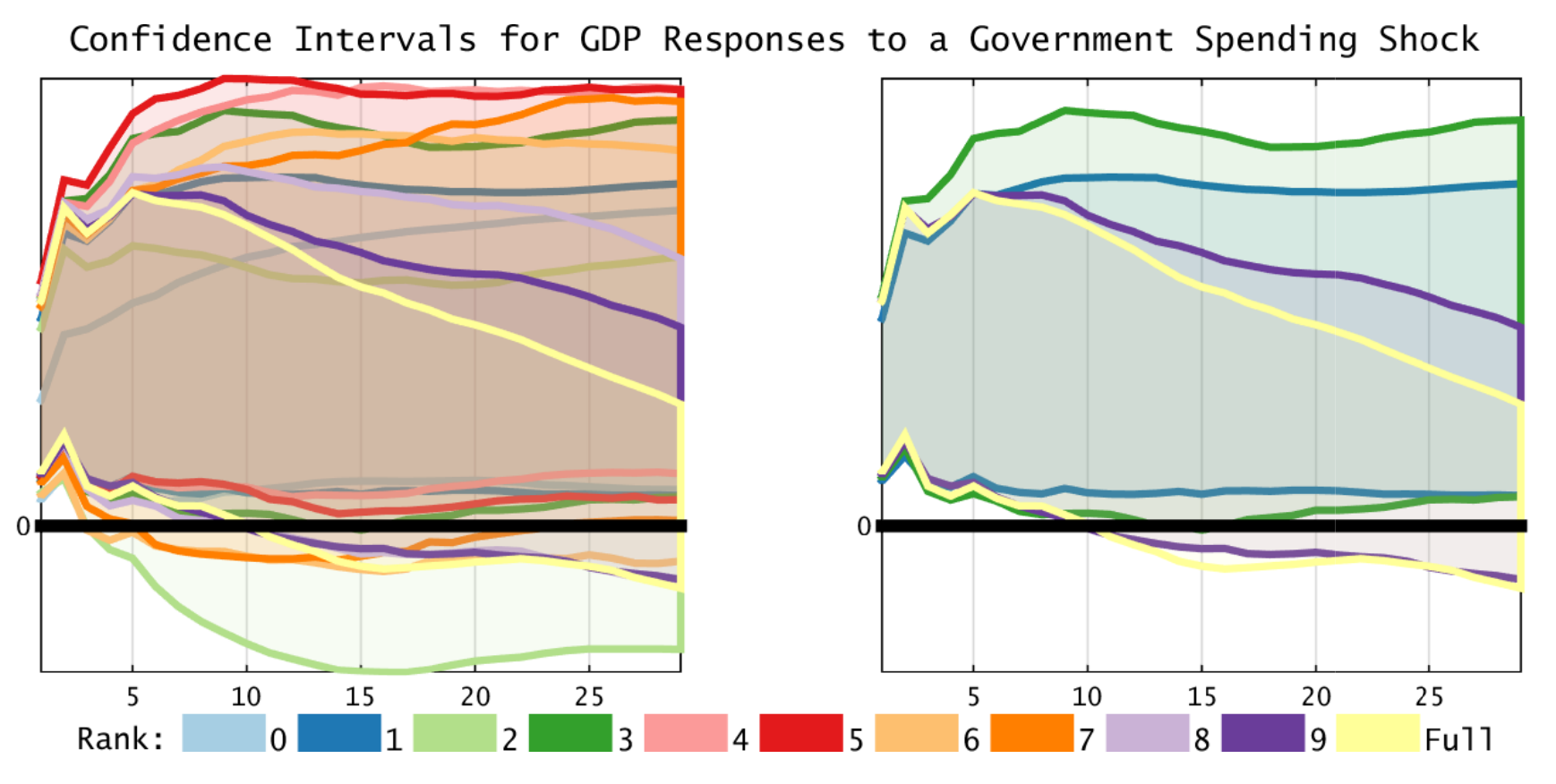}
\caption{Left panel: Bootstrap 95\% confidence intervals of the output response to a government spending shock for every rank specification. Right panel: Bootstrap 95\% confidence intervals of the output response to a government spending shock implied by the trace test ($r=3$), AIC ($r=9$), BIC ($r=1$), and the unrestricted VAR}
\label{Example_Intro}
\end{center}
\end{figure}

\section{Inference Accounting for Rank Uncertainty} \label{sec:r_unc}

In this section we discuss several ways of accounting for rank uncertainty, first utilizing existing methods from the model uncertainty literature, before discussing a new principle.

\subsection{Adaptations of Existing Model Uncertainty Methods} \label{sec:ada}
The perils of ignoring model uncertainty when performing model selection are well known in the statistical literature about model selection. For instance, in a sequence of papers, Leeb and P\"otscher \citep[see for example][]{LeebPoetscher05} highlight the risk of treating a selected model as a known and correct when performing inference, pointing out that even consistent model selection is no justification for treating the selected model as known. While this post-model selection inference problem is hard to solve, various methods have been proposed to at least mitigate the problem. Here we highlight some of these methods and show how they can be adapted to the problem at hand. We stress though that, although they are regularly used in practice to account for model uncertainty, none of these methods are formally shown to deliver valid post-model selection inference.

The most straightforward way, and our baseline benchmark, to deal with rank uncertainty is to pre-estimate the rank, and then perform inference for the impulse responses conditional on the estimated rank. While this seems, given the discussion in the previous section, not always an advisable strategy, rank estimation underlies many of the methods considered afterwards. We therefore first discuss how to perform rank estimation and how it can be seen as a model selection problem.

Let the function $M_r(Y_T): Y_T \mapsto \{0,1\ldots,K\}$ be a rank selection procedure that determines the cointegration rank based on the sample $Y_T = (y_1, \ldots, y_T)^\prime$. Then the estimated rank $\hat{r}$ can be imposed in the VECM estimation to obtain the estimated impulse responses of interest as $\hat{\zeta}^{(\hat{r})} = \bar{f} (\hat{\theta}^{(\hat{r})})$,where $\hat{r} = M_r (Y_T)$.

Several methods can be considered in practice for estimation of the rank. The most common is to perform a sequence of sequential tests in the likelihood framework of \citet{Johansen95}, in particular using the trace or eigenvalue test statistics. Instead of the standard critical values, one can also use one of its many bootstrap extensions \citep{CRT10ET,CRT10JoE,CRT12,Swensen06}. Either way, due to the nature of hypothesis testing, this estimation strategy will not lead to consistent estimation of the rank (unless the significance level is chosen to decrease with sample size); the probability of selecting a rank that is too high converges to the chosen significance level instead of to zero.

Alternatively, one can use an information criterion as proposed by \citet{Phillips96}, \citet{ChaoPhillips99}, \citet{ChengPhillips09} and \citet{ChengPhillips12}. This has two advantages compared to the sequential testing approach. First, rank selection and lag length selection can be done in a single step. Second, depending on the penalty function chosen in the information criterion, it is possible to estimate the rank consistently. A recent alternative is provided by \citet{LiaoPhillips15} who propose to select the rank and lag length simultaneously by penalized reduced rank regression. An advantage of this approach is that model selection and estimation are performed simultaneously, thus needing only a single step for the full estimation from start to end.

Irrespective of the chosen selection method, standard inference is based on the selected rank, treating it as known. This is often justified by the consistency of the rank selection method, but even in those cases where it is indeed consistent, ignoring the selection step leads to invalid inference as referred to earlier \citep{LeebPoetscher05}. In particular if the data do not provide clear and strong evidence for one particular cointegrating rank, this approach will fail to deliver reliable confidence intervals. We therefore next consider methods that explicitly take rank uncertainty into account in the inference procedure.

\subsubsection{Endogenous Rank Selection} \label{sec:r_end}
\citet{Kilian98b} proposes the \textit{endogenous lag selection} bootstrap method for autoregressive models where the autoregressive lag length is re-estimated within the bootstrap to account for the model selection uncertainty. We adapt his approach to rank selection, labeling this approach \textit{Bootstrap Endogenous Rank Selection (BERS)}. That is, after generating a bootstrap sample $\{y_t^*\}_{t=1^T}$ with rank $r$, we re-estimate the rank from this bootstrap sample to estimate the bootstrap impulse responses.\footnote{Details for our implementation are given in Algorithm \ref{alg:b_ers}.}

We can choose to generate the bootstrap sample $\{y_t^*\}_{t=1^T}$ with the ``neutral'' maximum rank $K$ or the estimated rank $\hat{r}$. While \citet{Kilian98b} reports that this choice has little consequence for lag selection, this is very different for rank selection. After all, if the rank used to generate $\{y_t^*\}_{t=1^T}$ is not correct, we still face all the problems with the bootstrap as we described before. Hence, while some rank uncertainty is taken into account, the validity of this approach still hinges on the correct rank being used for the generation of the bootstrap data, which as we argued before, is impossible to guarantee.

\subsubsection{Model Averaging} \label{sec:ma}
One of the most popular approaches to account for model uncertainty is to use model averaging \citep{HjortClaeskens03}. By combining estimators from different models (and potentially weighting by evidence for these models), model uncertainty is taken into account. Given that the decision of which model to use is discrete, and therefore the selected model may change abruptly for a slight variation in the sample, the resulting estimators after model selection may be quite unstable and exhibit a large variability. By constructing weighted averages of the estimators arising from the individual models, one smoothes out the changes in the estimator, resulting in more stable estimators that typically display lower variability.

Given rank-specific impulse response estimators $\hat{\zeta}^{(0)}, \ldots, \hat{\zeta}^{(K)}$, we define the \textit{Model Averaging (MA)} impulse response estimator 
\begin{equation} \label{eq:est_w}
\hat{\zeta}^{MA} = \sum_{r=0}^K W_{K} (r) \hat{\zeta}^{(r)}, \qquad \text{where} \qquad W_{K} (r) = \frac{ W(Y_T,r)}{\sum_{s=0}^K W(Y_T,s)}
\end{equation}
and $W(Y_T,r)$ is a function that determines a weight for rank $r$ based on the sample $Y_T$. Unlike the typical application of model averaging, which often focuses on improving accuracy of point estimators in a mean squared error sense, we are not interested in the averaged point estimators. Instead, we only take the MA estimator as an input into our bootstrap scheme in order to construct confidence intervals: By using the more stable MA estimator, we may hope that the confidence intervals are more robust to rank misspecification. The bootstrap scheme can straightforwardly be adapted to incorporate this estimator after generating the bootstrap sample $\{y_t^*\}_{t=1}^T$.

Typical weights in the model averaging literature are exponential weights based on information criteria such as BIC. However, in our simulations we find that such standard weighting schemes give weights that are too close to each other and do not differ much from simple unweighted averages. Given the widely varying behavior of impulse responses under different ranks, such weights are therefore not the most useful ones in our setting. Instead, we advocate using weights that are derived directly from cointegration tests, following the spirit of \citet{SobreiraNunes12}, but rather than their KPSS type weights, we opt for weights based on the trace test statistic proposed by \citet{Johansen95}. Details about the weights and their properties can be found in Lemma \ref{lem:weights} in Section \ref{sec:wimp_av}.

In a similar framework, \citet{JMP13} propose an averaging approach for impulse responses of potentially cointegrated VAR models based on a very specific set of weights. While they allow for uncertainty regarding the order of integration, their approach only averages two estimators: the one obtained from the VAR in levels, and one obtained from a cointegrated VAR where the number of cointegrating relations is pre-determined by pre-testing or economic theory. It can therefore not account for the general case where we are agnostic about the number of cointegration relations.

While such model averaging explicitly takes model uncertainty into account, it still relies on an explicit choice of the cointegration rank in the bootstrap algorithm to do inference. Hence, even while the weight construction can be endogenized in the bootstrap in the same way as for rank selection, the bootstrap DGP relies on the choice of a single cointegration rank. As such it still does not fully account for rank uncertainty in our context.

\subsubsection{Bagging} \label{sec:bag}
We now take a first step in endogenizing the rank uncertainty in the bootstrap DGP itself, by bootstrapping a bagging estimator. The bagging estimator is constructed by averaging the bootstrap estimates over an initial bootstrap procedure in which the cointegration rank is re-estimated for every bootstrap sample. Bagging was originally proposed by \citet{Breiman96} to improve estimation accuracy of unstable estimators. \citet{BuhlmannYu02} analyzed bagging formally and found that it can lead to a variance reduction of estimation after hard decisions, such as an initial model selection. As the model averaging described above, bagging smoothes those hard decisions yielding more accurate estimators. \citet{Efron14} considers bagging in the context of post-selection inference, rather than point estimation, and we build on his approach here.

As bagging is essentially the simulation equivalent of model averaging, with the weights implicitly determined by how often each rank is selected within the bootstrap, it is subject to the same critique. However, one can modify the bagging algorithm to endogenize rank uncertainty in the bootstrap DGP by performing a second-level bootstrap in which we draw new bootstrap samples from the first-level bootstrap samples. By determining the rank of the second-level bootstrap DGPs from the first-level bootstrap samples, the ranks are randomized according to their evidence in the (simulated) sample. This allows to take the uncertainty into account when constructing the bootstrap confidence intervals based on the second-level bootstrap samples. While this does not fully solve the bootstrap invalidity problem (bootstrap samples are still generated under incorrect ranks, especially in the first step), the method has the potential to alleviate the problem.

There is a computational problem with this method though, as one has $B_1$ iterations in the first bootstrap and $B_2$ in each second-level bootstrap, such that a full double bootstrap requires $B_1 (1+B_2)$ iterations which quickly becomes computationally infeasible. To circumvent this problem, we implement the Fast Double Bootstrap (FDB) developed by \citet{DavidsonMacKinnon02}, which requires drawing only a single second-level bootstrap sample for every first-level bootstrap sample. That is, the computation cost of the FDB is only double ($2 B_1$) that of a regular bootstrap. Algorithm \ref{alg:b_fdb} describes the method, labeled as \textit{FDB bagging (FDBb)}, in detail.

\subsection{Weighted Inference by Model Plausibility} \label{sec:wimp}
None of the methods described above fully address the post-model selection inference problem. To work towards a more satisfactory solution, we now combine the ideas discussed above with new concepts arising from the recent statistical literature that directly addresses the post-model selection inference problem.

We would like to build on the idea of averaging or weighting models to account for rank uncertainty. However, as elaborated on in the previous section, such weighting is typically designed for point estimation and translating it to confidence intervals, as needed here, is not straightforward. In order to make the transition, we take inspiration from the perspective taken by \citet{PoSI13}, who view the issue of constructing valid post-model selection inference (PoSI) as a simultaneous inference problem: by controlling for performing inference in all models simultaneously, the specific model selected by a model selection procedure is covered by construction. This would involve finding lower and upper bounds $L^{\PoSI} (\gamma)$ and $U^{\PoSI}(\gamma)$ to construct intervals $\left[L^{\PoSI} (\gamma), U^{\PoSI} (\gamma)\right]$ such that $\Prob \left( L^{\PoSI} (\gamma) \leq \zeta^{(r)} \leq U^{\PoSI} (\gamma), \quad \forall r \in\{0,1,\ldots,K\} \right) \rightarrow 1 - \gamma$ as $T \rightarrow \infty$. Note that $\zeta^{(r)} = \bar{f} (\theta^{(r)})$ is a \emph{pseudo-true} parameter defined in terms of $\theta^{(r)}$, the pseudo-true parameters of the model \eqref{eq:VAR} under the restriction that rank $r$ is imposed -- see Lemma 1 and its proof in \citet{CRT12} for a formal definition. These parameters represent the probability limits of the estimators of \eqref{eq:VAR} under the restriction of imposing rank $r$, and can informally be seen as those parameters which minimize a distance to the true parameters under the restriction that the cointegration rank is $r$. If $r < r_0$, the true parameter cannot be recovered, and therefore the pseudo-true parameter will be different.

For our purposes, there is a fundamental problem with the \emph{sub-model} view of \citet{PoSI13} where the pseudo-true parameters are the objects of interests, as also highlighted by \citet{LPE15}. In the context of structural impulse responses, the sub-model view has little relevance, as it cannot uncover any structural effects. We therefore need the \emph{full model} view, in which it is assumed that one of the models is the true (structural) one. Denoting this extended PoSI approach as $\PoSI_0$, we seek to control $\Prob \left( L^{\PoSI_0} (\gamma) \leq \zeta \leq U^{\PoSI_0} (\gamma), \quad \forall r \in\{0,1,\ldots,K\} \right) \rightarrow 1 - \gamma$ as $T \rightarrow \infty$. As the interval bounds are typically constructed by considering the distribution of the fixed-rank estimator $\hat{\zeta}^{(r)}$ minus the (pseudo-)true value, this approach requires that the distance between every fixed-rank estimate $\hat{\zeta}^{(r)}$ and the true impulse response $\zeta$ is accounted for, rather than the much shorter distance between $\hat{\zeta}^{(r)}$ and its probability limit or pseudo-true impulse response $\zeta^{(r)}$. This will therefore result in rather wide intervals. The seemingly only way to control this quantity is to construct confidence intervals for every rank separately, and then take the union of these, which typically results in very wide intervals that are useless in practice.

However, we have not yet considered any evidence on the plausibility of each rank, that can be extracted from the data. If this information can incorporated into our inferential procedure, we may be able to achieve intervals that are still useful in applications, as the impact of ranks that the data deem very implausible can be eliminated, or at least reduced. We therefore augment the PoSI view of simultaneous inference by a weighting scheme akin to model averaging, except that we apply the weighting not to the estimators but directly to the bounds of the intervals. The direct weighting of the inference output, in this case the interval bounds, by evidence of the plausibility of each model, leads us to label our approach as \textit{Weighted Inference by Model Plausibility (WIMP)}.

\subsubsection{The WIMP Principle} \label{sec:wimp_p}

Define the most plausible model - according to a certain plausibility measure based on the data - as the \emph{reference model}, and denote the corresponding confidence interval arising from this model (ignoring model uncertainty) as the \emph{reference interval}. As input to the WIMP procedure we consider all \emph{model intervals}, which are defined as the confidence intervals obtained by assuming any particular model as the true one. In our case these would be the intervals obtained by imposing all the $K+1$ different cointegrating ranks. Before going into the details of our application, we now propose a set of general conditions that a ``prudent'' WIMP scheme should adhere to: \medskip \\
\noindent \textbf{WIMP Prudence Conditions}
\begin{enumerate}
\item The WIMP confidence interval must always cover at least the reference interval. That is, any non-reference model can only lead to widening the WIMP interval compared to the reference interval.
\item If two models are equally plausible, the model interval bounds which are furthest away from the reference model must contribute the most to widening the WIMP interval.
\item If the bounds of two model intervals are equally far away from the reference interval, the most plausible model must contribute the most to widening the WIMP interval for a given distance of the bounds from the reference interval.
\item The WIMP confidence interval may not be wider than the interval obtained by joining all individual model intervals.
\end{enumerate}

\medskip \noindent
The first condition is needed to avoid invalid intervals, in whatever way validity is measured. If obtaining a confidence interval which is more narrow than the ``standard'' interval assuming no model uncertainty is possible, the WIMP interval is unlikely to contain an adequate coverage probability. The second condition ensures that the locations of intervals in relation to the reference interval are properly taken into account for equally plausible models. Compare two equally plausible models with almost identical intervals, to two equally plausible models with very different intervals. Any prudent method of accounting for model uncertainty must result in wider intervals for the second case than for the first case. The third condition implies that plausible models are more strongly taken into account than implausible models. In particular, this condition allows to reduce the impact of implausible models that may have very different intervals than the reference model but are so implausible, that there is little to no uncertainty about them. Finally, the fourth condition ensures that the WIMP intervals do not become too conservative. While the first and fourth condition impose hard (but sensible) restrictions on the WIMP intervals, the second and third conditions allow for variation in the procedure. Finding a right balance between conservatism and interval length is therefore of great practical importance, and varies per setting.

For our specific implementation of the WIMP Prudence Conditions, let $W_{K}(r)$ be model plausibility weights assigned to all ranks $r=0,\ldots,K$ and define $X(r,s) = \frac{W_K(r)}{W_K(s)}$ as the relative plausibility of rank $r$ compared to rank $s$. Letting $R = \argmax_{0 \leq r \leq K} W_K(r)$ be the (most plausible) reference rank, we define the WIMP interval $\left[L^{\WIMP}(\gamma), U^{\WIMP} (\gamma) \right]$ as
\begin{equation} \label{eq:ci_sposi}
\begin{split}
L^{\WIMP}(\gamma) &= \min_{r=0,\ldots,K} \left\{L^{(R)} (\gamma) - X(r,R) \left[ L^{(r)}(\gamma) - L^{(R)}(\gamma) \right]^- \right\},\\
U^{\WIMP}(\gamma) &= \max_{r=0,\ldots,K} \left\{U^{(R)} (\gamma)+ X(r,R) \left[U^{(r)}(\gamma) - U^{(R)}(\gamma) \right]^+ \right\},
\end{split}
\end{equation}
where $x^+ = \max(x,0)$, $x^- = -\min(x,0)$ and $L^{(r)}(\gamma)$ and $U^{(r)}(\gamma)$ are the lower and upper bounds respectively of the confidence intervals with fixed rank $r$.

The term $\left[ L^{(r)}(\gamma) - L^{(R)}(\gamma) \right]^-$ (respectively $\left[U^{(r)}(\gamma) - U^{(R)}(\gamma) \right]^+$) ensures that only lower bounds smaller (upper bounds larger) than those of the reference interval are taken into account; for lower bounds larger (upper bounds smaller) than those of the reference interval, this term is simply zero. Together with $X(r,s) \geq 0$, this implies that the WIMP interval always contains the reference interval, hence Condition 1 is satisfied. Condition 2 is also trivially satisfied as this term increases when the lower (upper) bound of the rank $r$ interval is further away from the reference interval.

The shape of $X(r,s)$ determines how strongly less plausible models are taken into account and can be different from the linear function of $W_K(r)$ imposed above. As long as $X(r,s)$ is an increasing function of $W_K(r)$, more plausible ranks are given more importance and Condition 3 is satisfied; varying $X(r,s)$ and $W_K(r)$ allows one to change the balance between conservatism and interval length. Finally, with respect to Condition 4, note that as long as $X(r,s) \leq 1$, the WIMP interval can never be wider than the interval obtained by combining the smallest lower bound with the largest upper bound.\footnote{If some of the individual model intervals are disjoint, the ``maximal'' WIMP interval as constructed in \eqref{eq:ci_sposi} is larger than the union of these intervals, apparently violating Condition 4. It is a matter of personal preference whether to consider disjoint intervals or to ``fill the gaps'' and extend it from the lowest lower bound to the highest upper bound, which is exactly what the WIMP construction described above does automatically. As we believe that such a disjointed confidence \emph{set}, which is not a confidence \emph{interval} anymore, can be rather difficult to interpret, we consider this modification, though it is by no means crucial to the WIMP approach.}

\begin{remark}
Although we focus here exclusively on the case of rank uncertainty, other types as uncertainty, such as about the lag order or the deterministic components can be incorporated into the WIMP procedure as well. For instance, if one wants to allow for $P$ different lag orders in addition to the $K+1$ ranks, one needs weights that measure the plausibility of each of the $(K+1)P$ different models resulting from combining the different ranks and lag orders. In this paper we focus on rank uncertainty only as it has a far bigger and more fundamental impact than (slight) lag misspecification. Uncertainty about the deterministic specification is typically a bigger issue, but due to our initial detrending all consequent analysis (including the statistics used to construct $W_K(r)$) are invariant to the deterministic specification (also see Remark \ref{rem:detr}), and we can separate the two sources of uncertainty.
\end{remark}

\begin{remark}
The WIMP intervals are not built directly around a single point estimator for $\zeta$. While all $K+1$ fixed-rank estimators are incorporated through their respective confidence intervals, we do not directly obtain a corresponding point estimate for $\zeta$. Of course, if there is a desire to pair the confidence interval with a point estimator, one can do so, in which case the model averaging estimator with the same weights $W_K(\cdot)$ as used for the WIMP intervals is the most natural candidate.\footnote{As expected from the model averaging literature, unreported simulations in the same setup as considered in Section \ref{sec:sim} show that this estimator performs very well in terms of mean squared error when compared to fixed-rank estimators. Of course, its performance purely as a point estimator is different from its performance as basis for inference, as we shall see in Section \ref{sec:sim}.}
\end{remark}

\subsubsection{Asymptotic Properties} \label{sec:wimp_av}
To complete our theoretical discussion of the WIMP method, we establish some basic asymptotic properties of the WIMP intervals. We mainly do so under general high-level assumptions on the tests and bootstrap method available, but we will also provide some details about how these assumptions can be verified in our application. We first characterize the general asymptotic properties of our method.

\begin{theorem} \label{th:av_pw}
Let $Y_T$ be generated according to \eqref{eq:VAR}, and let $\Theta^{(r)}$ denote the parameter space of $\theta$ such that the $I(1,r)$ conditions are satisfied. Then assume that
\begin{enumerate}[(i)]
\item As $T \rightarrow \infty$, $\Prob \left(W_{K}(r_0) \geq W_K(r)\right) \rightarrow 1$ for all $r \neq r_0$;
\item As $T\rightarrow \infty$, it holds that
\vspace{-0.5cm}
\begin{equation*}
\Prob \left(L^{(r_0)} (\gamma) \leq \zeta \leq U^{(r_0)}(\gamma) \right) \rightarrow 1 - \gamma, \quad \text{for all } \theta \in \Theta^{(r_0)} \text{ and } r_0 \in \{0,1,\ldots,K\}.
\end{equation*}
\end{enumerate}
\vspace{-0.5cm}
Then, as $T \rightarrow \infty$,
\vspace{-0.5cm}
\begin{equation*}
\Prob \left( L^{\WIMP} (\gamma) \leq \zeta \leq U^{\WIMP} (\gamma) \right) \geq 1 - \gamma +o(1), \quad \text{for all } \theta \in \Theta^{(r_0)} \text{ and } r_0 \in \{0,1,\ldots,K\}.
\end{equation*}
\end{theorem}

Theorem \ref{th:av_pw} establishes the asymptotic conservativeness of the WIMP intervals under two assumptions. First, (i) requires that the weight attached to the true rank is asymptotically at least as large as the weight of the other ranks. This requires that a ``decent'', yet not necessarily consistent, procedure is used to obtain the weights. Equal weights satisfy this condition, but will lead to too conservative intervals as they would result in taking the union of all rank $r$ intervals. Note that if condition (i) is strengthened to require the true weight to receive the full weight asymptotically, corresponding to using a consistent rank selection approach, the WIMP interval is not conservative anymore but has the appropriate (pointwise) coverage rate.

Assumption (ii) implies pointwise asymptotic validity of the intervals under a known rank, which has been verified for many bootstrap methods under different assumptions on $\{u_t\}$ (or equivalently $\{\varepsilon_t\}$). For instance, if we assume that $\{u_t\}$ is i.i.d.~with sufficiently many moments existing, one can show that the i.i.d.~bootstrap version of Algorithm \ref{alg:b} satisfies assumption (ii), c.f.~\citet{Kilian98} and \citet{CRT12}. \citet{InoueKilian16} also formulate general assumptions to assure bootstrap validity, while alternative methods that allow for heteroskedasticity are considered by \citet{BJT16}. The WIMP principle can be applied to any of these - or other - methods.

We now propose a simple weighting scheme and consider its asymptotic properties. Following the spirit of \citet{SobreiraNunes12}, we base our weights on cointegration tests. Rather than their KPSS type weights, we opt for weights based on the trace test statistic proposed by \citet{Johansen95}, which, as a ``standard'' cointegration test, has intuitive appeal and is available in all standard econometric and statistical software.\footnote{We also explored \citeauthor{Johansen95}'s \citeyearpar{Johansen95} maximum eigenvalue test statistic, which similarly satisfies assumption (i) in Theorem \ref{th:av_pw}. Numerical experiments showed virtually no difference with the trace test.}

\begin{proposition} \label{lem:weights}
Let $J_T(r) = -T \sum_{i=r+1}^K \ln (1 - \hat{\lambda}_i)$ denote the trace test of \citet{Johansen95} for testing $H_0: r_0 \leq r$. For constants $c_1>0$ and $0<c_2<1$, define
\begin{equation} \label{eq:weights}
\begin{array}{ll}
W(Y_T,r) = e^{-c_1 T^{-c_2} J_T(r)} &\quad \text{for } r = 0.\\
W(Y_T,r) = e^{-c_1 T^{-c_2} J_T(r)} - e^{-c_1 T^{-c_2} J_T(r-1)} &\quad \text{for } r=1,\ldots, K-1, \\
W(Y_T,r) = 1-e^{-c_1 T^{-c_2} J_T(r-1)} &\quad \text{for } r = K,
\end{array}
\end{equation}
and $W_{K} (r) = W(Y_T,r)/ \sum_{r=0}^K W(Y_T,r)$. Then $W_{K}(r) \xrightarrow{p} \mathbbm{1} (r=r_0)$ as $T \rightarrow \infty$.
\end{proposition}

Note that our weights ensure that the true rank asymptotically receives a weight of one, which is stronger than required in Assumption (i). This implies that using these weights, the WIMP intervals are not conservative asymptotically. In practice one faces the trade-off between the desired robustness to model uncertainty and the width of the resulting intervals. However, we stress that changing the constants $c_1$ and $c_2$, may lead to very different small sample properties -- even though the asymptotic properties remain unaffected. Therefore it remains crucial to investigate the small sample properties of the chosen approach. This we do in the next sections.

\begin{remark}
While the results above establish pointwise asymptotic validity, this does not imply validity uniformly over the parameter space.\footnote{Note that our notion of uniform and pointwise validity is conceptually different from the notion occasionally encountered in the impulse response literature, such as \citet{LSW15} and \citet{InoueKilian16}. In those papers, ``pointwise'' relates to inference on a single impulse response, whereas uniform or joint confidence bands are valid for a set of impulse responses. Our notion of uniform and pointwise relates to to the parameter space $\Theta$, and applies to both inference on single responses and joint inference on a set of responses. Methods establishing joint coverage are as sensitive to rank uncertainty as methods for single impulse responses, and our arguments apply equally well to these methods.}

Uniform validity is a more informative property about finite sample behavior of the intervals, as it explicitly accounts for ``small'' parameters, such as roots that are local to one. While one may expect that Assumption (i) helps to establish uniform validity by not relying on the \textit{oracle property} that the true rank is always selected asymptotically, one would also need a uniform version of Assumption (ii). While this would allow us to formulate Theorem \ref{th:av_pw} in a uniform sense, we do not do so as it would hide the fact that the current state of the (bootstrap) literature does not provide general bootstrap approaches that are uniformly valid in a general sense. To the best of our knowledge, uniform results have only been established in the presence of a single local-to-unit root \citep[cf.][]{Mikusheva07,Mikusheva12}, while our setting would require validity under an \emph{arbitrary} number of roots near unity. While clearly of great interest, developing appropriate bootstrap methods would require a separate study that is outside the scope of the paper and is therefore left for future research.
\end{remark}

\section{Monte Carlo Simulations} \label{sec:sim}
In this section we investigate the performance of the various methods discussed above by simulation. We assess coverage probabilities (CP) of confidence bands for \textit{forecast error impulse responses}, and hence evaluate intervals for the moving average parameters. We intentionally abstract from the identification problem in structural VARs, since the structural moving average parameters are linear combinations of their reduced-form counterparts, and one can expect that the performance of one inferential procedure for reduced-form parameters is inherited by the structural parameters.\footnote{Except for SVARs identified through long-run restrictions, the exact persistence properties of the underlying reduced-form process are of no direct relevance for identification.}
The data generating process (DGP) for the Monte Carlo experiment is a three-dimensional VAR of order one inspired by \citet{Phillips98}, given by $y_t = (I_{3} +\Pi)y_{t-1}+\epsilon_t$, with $\epsilon_t \sim i.i.d.~\mathcal{N}(0,I_3)$ for all $t$. The cointegration matrix is specified as $\Pi=d_1 \alpha_1\beta_1' + d_2\alpha_2\beta_2'$, where $\alpha_1=(0,1,0)'$, $\alpha_2=(0,0,1)'$, $\beta_1=(2,-1,0)'$,  and $\beta_2=(1,-1,-1)'$. We consider two versions of the above process when simulating data.
\textbf{DGP1} features two ``\textit{weak}'' cointegration relations by setting $d_1=0.05$ and $d_2=0.02$, which implies that the model has one root at unity and two roots close to one at $0.98$ and $0.95$. \textbf{DGP2} features two ``\textit{strong}'' cointegration relations by setting $d_1=d_2=1$, which implies a VAR with one unit root and two roots at zero. This is the original setting considered by \citet{Phillips98}.

We evaluate CPs of $95\%$ confidence intervals for each response and horizon ($h=1,2,...,60$) for $T=100, 200$. The results are based on 1000 MC simulations and 399 bootstrap replications. To compute the WIMP intervals we set $c_1 = 1$ and $c_2=0.5$ for the weights in \eqref{eq:weights}.\footnote{This choice of parameters seems to be natural for the weights in \eqref{eq:weights}. We did not experiment with changing these values, as the performance in the simulations was already quite satisfactory. It is likely that by careful tuning these parameters, even better performance can be obtained. However, the optimal choice will typically be highly case-dependent, and optimal values should therefore be treated with caution. Instead we prefer to report results for a natural albeit naive choice of parameters without claiming any optimality.} We abstract from lag length selection (we fix $p=1$), deterministic components, and small sample bias correction \citep{Kilian98}. All simulations were done in MATLAB.

\begin{figure}
\begin{center}
\includegraphics[width=1\linewidth]{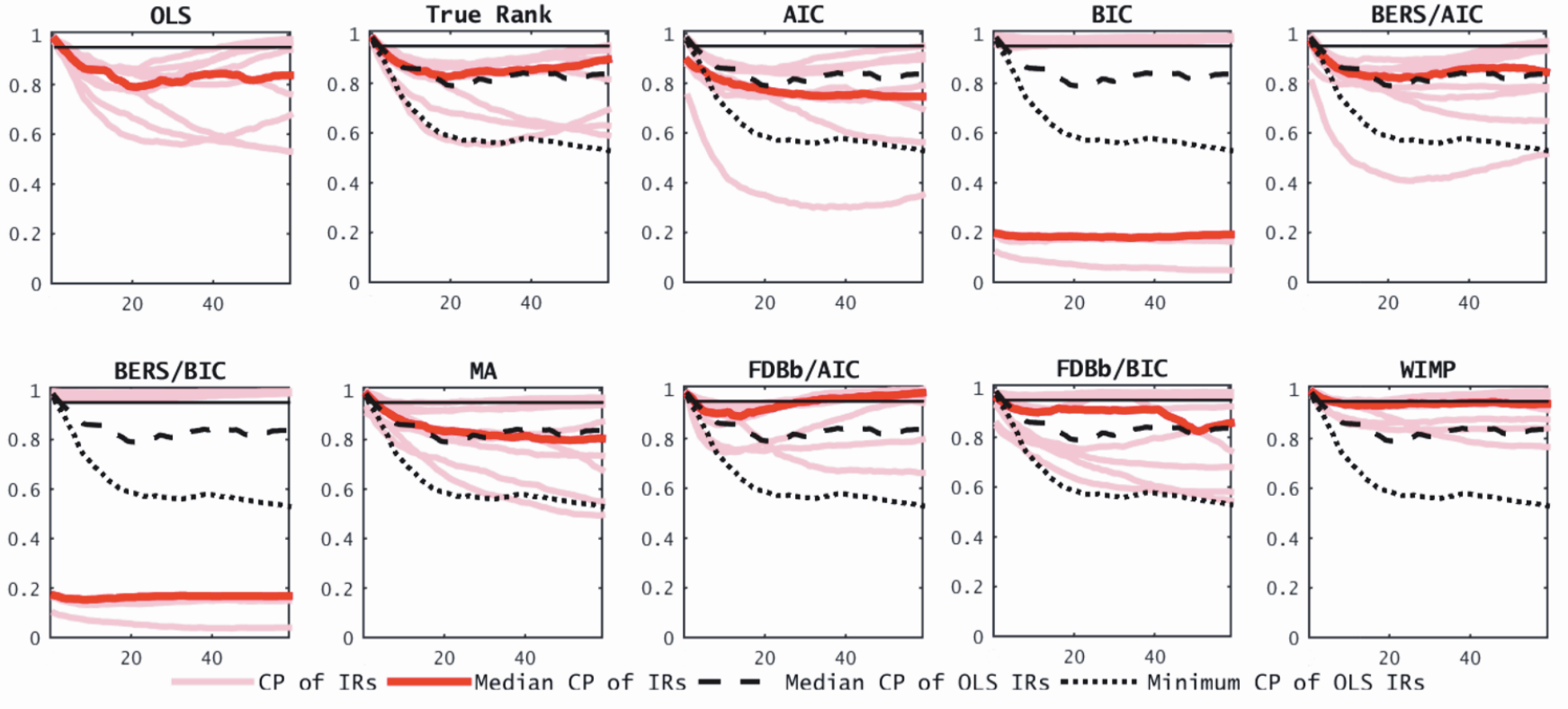}
\caption{DGP1: Empirical coverage rates for $T=100$. 
`\textbf{OLS}': (unrestricted) VAR in levels estimated by OLS; `\textbf{True Rank}': VECM estimated with knowledge of the true rank; `\textbf{AIC}' and`\textbf{BIC}': rank estimation using AIC and BIC, respectively; `\textbf{BERS/AIC}' and `\textbf{BERS/BIC}': Bootstrap Endogenous Rank Selection with respectively AIC and BIC used for rank selection; `\textbf{MA}' : Model Averaging with weights as in \eqref{eq:weights}; `\textbf{FDBb/AIC}' and `\textbf{FDBb/BIC}': FDB bagging with respectively AIC and BIC used for rank selection; `\textbf{WIMP}': WIMP method with weights as in \eqref{eq:weights}.
The {\bf\color{Pink2}{pink}} lines show CPs for all nine impulse responses; the {\bf\color{red}{red}} line is the median of these per horizon. For ease of comparison, the median and minimum coverage of the OLS intervals is always reported in \textbf{black}.}
\label{cp_dgp1_100}
\vspace{1cm}
\includegraphics[width=1\linewidth]{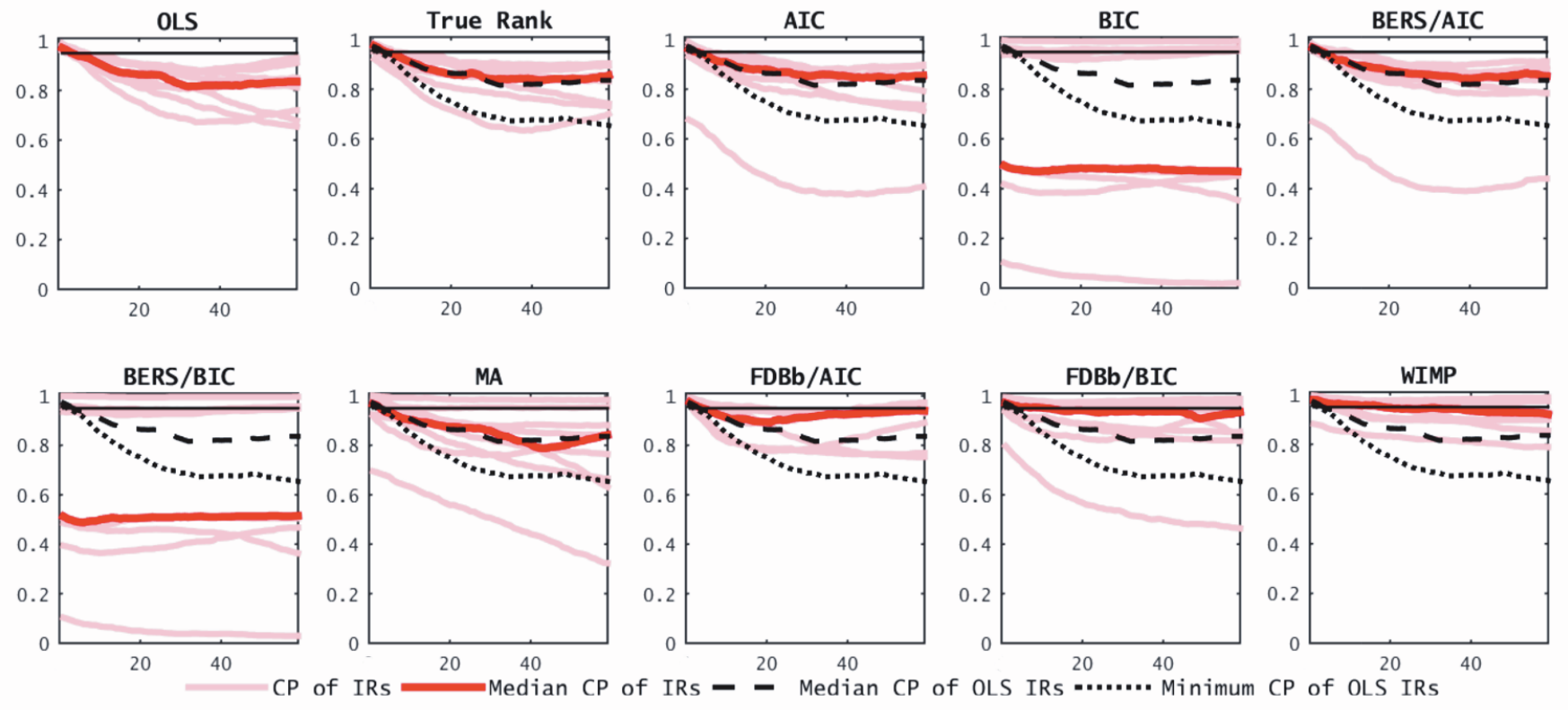}
\caption{DGP1: Empirical coverage rates for $T=200$. See Figure \ref{cp_dgp1_100} for details.}
\label{cp_dgp1_200}
\end{center}
\end{figure}
Figure \ref{cp_dgp1_100} and \ref{cp_dgp1_200} display CPs of the various inferential procedures discussed above for DGP1 for $T=100$ and $T=200$. Based on the two model selection criteria employed, we can partly confirm the findings of \citet{GHP13}. That is, if evidence for a particular rank is weak, pre-testing seems not to deliver more accurate inference than (bootstrap) CIs based on unrestricted OLS. This holds for both sample sizes considered. However, these two frequently used approaches can both not be considered as reliable strategies for the construction of inference -- minimum CPs are well below 60\%. Surprisingly, even when the true model specification is imposed (which could be considered to be the \textit{oracle} method), CPs are generally not closer to the nominal level either; both for short and long horizons. Endogenous rank selection does not seem to improve the performance compared to the pre-testing procedure. FDB bagging does give CPs closer to nominal level, in particular when based on AIC. However, the WIMP intervals outperform all other methods, and deliver CPs that are on average quite close to the $95\%$ nominal level.

Figure \ref{width_dgp1} presents the corresponding average width of the bootstrap intervals over all horizons for the five most relevant methods. There are several interesting observations to make from this figure. First, note that even though FDB bagging and WIMP produce much more accurate intervals than OLS or imposing the true rank, they actually do not produce intervals that are much wider and overly conservative. Second, even though the WIMP method produces more accurate intervals than FDB bagging, intervals are not wider, indicating that the mechanism imposed in the WIMP to reduce the impact of implausible models works well in practice.

\begin{figure}
\begin{center}
\includegraphics[width=1\linewidth]{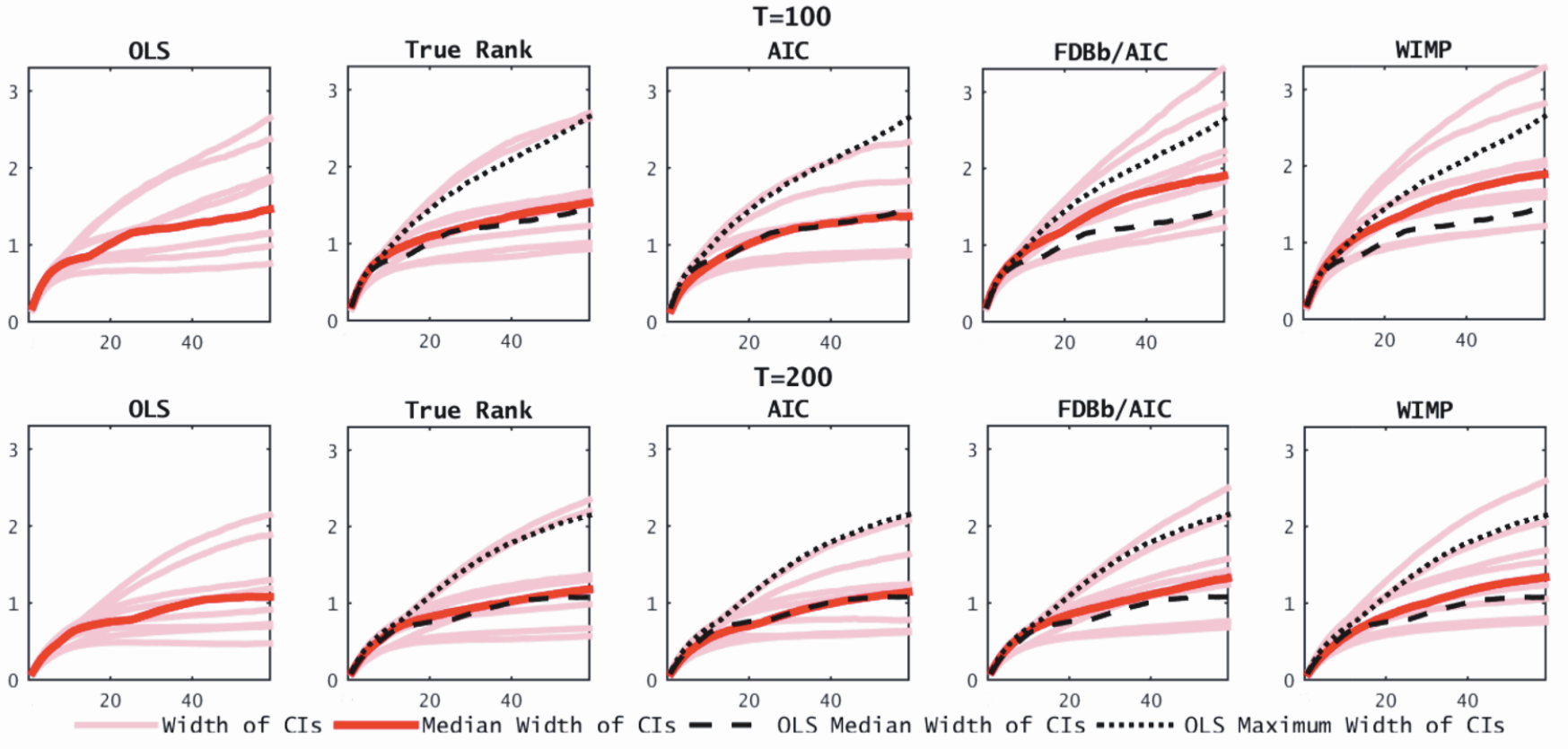}
\caption{DGP1: Average width of 95\% bootstrap CIs for various inference methods for $T=100$ and $T=200$. For details see Figure \ref{cp_dgp1_100}.}
\label{width_dgp1}
\end{center}
\end{figure}

It stands to reason that if evidence for a specific cointegration relation is strong, rank pre-estimation could result in more reliable inference than unrestricted OLS and may outperform the WIMP intervals which -- despite weighting down implausible ranks -- are inherently more conservative. We investigate this further by turning to DGP2. Figure \ref{cp_dgp2} displays CPs for the case of strong cointegration relations. Indeed, CPs implied by model selection based on AIC and BIC are much closer to the nominal level than those entailed by OLS. Bootstrap intervals based on unrestricted estimation can again not be considered as reliable, with minimum CPs around 60\% for both sample sizes. Imposing the true rank delivers CPs close to but still below the nominal level. As in the weak cointegration setting, the WIMP intervals again outperform all other approaches and even deliver CPs closer to nominal level than those implied by the correct rank specification. It is noticeable that the WIMP intervals do not produce overly conservative inference when evidence for a particular rank is strong, but result in CPs very close to the 95\% level. This is also reflected in the average width (over 1000 MC simulations) of the CIs displayed in Figure (\ref{width_dgp2}). WIMP intervals are (if at all) only marginally wider than those implied by the correct rank specification, and are even much narrower than some of the intervals based on the unrestricted model. Finally, note that the WIMP intervals are now also much narrower than some of the FDB bagging intervals while having superior coverage. 

\begin{figure}
\begin{center}
\includegraphics[width=1\linewidth]{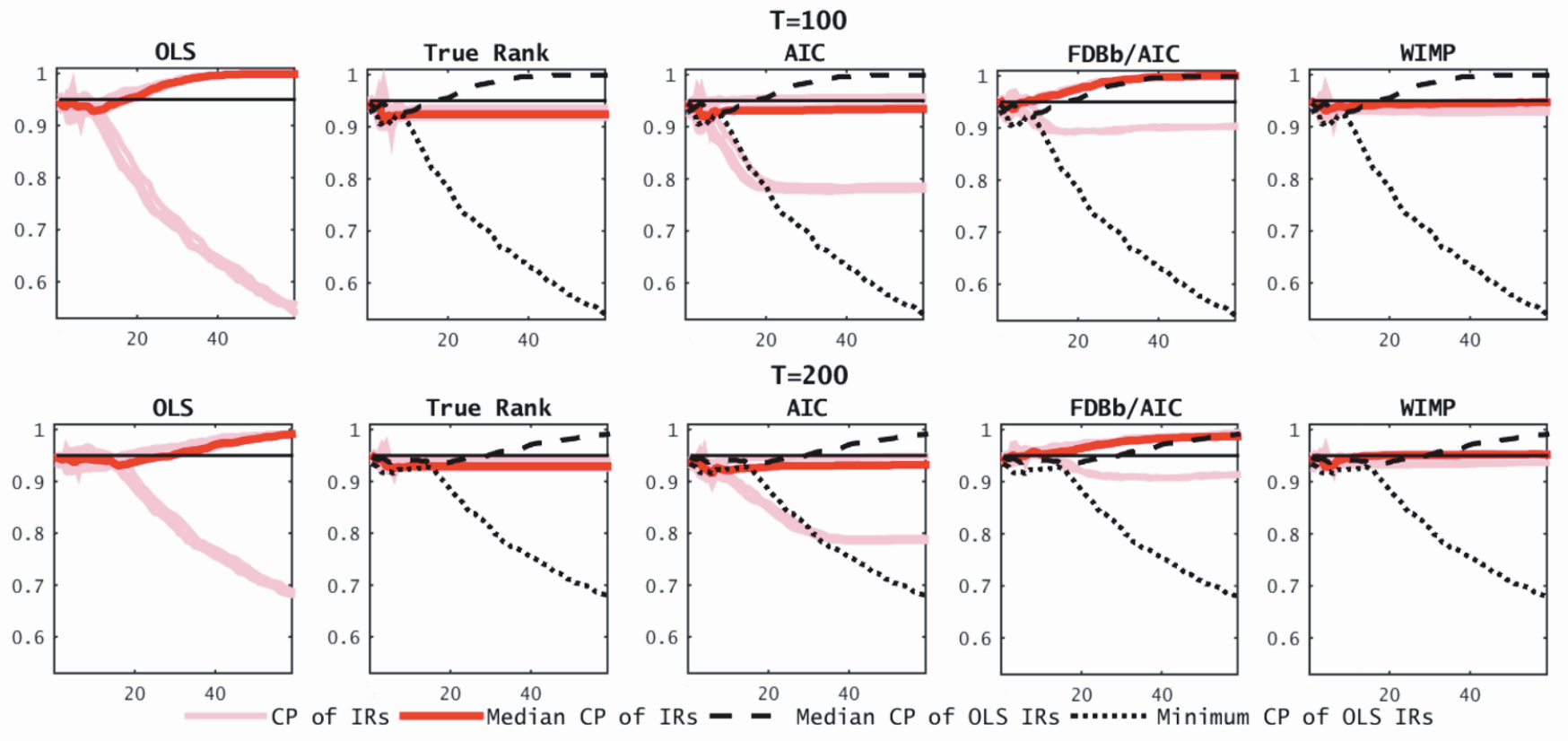}
\caption{DGP2: Empirical coverage rates for various inference methods for $T=100$ and $T=200$. For details see Figure \ref{cp_dgp1_100}.}
\label{cp_dgp2}
\end{center}
\end{figure}
\begin{figure}
\begin{center}
\includegraphics[width=1\linewidth]{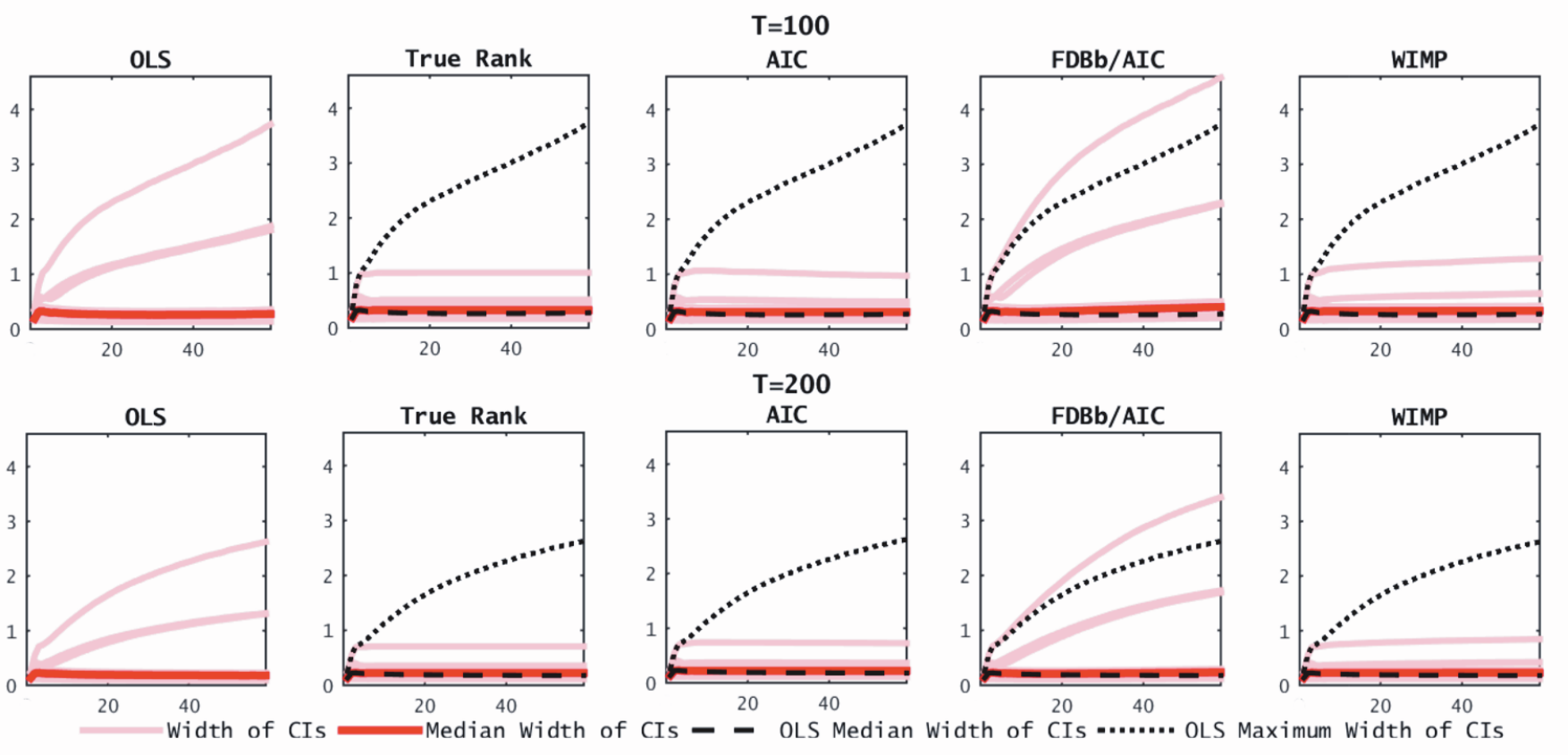}
\caption{DGP2: Average width of 95\% bootstrap CIs for various inference methods for $T=100$ and $T=200$. For details see Figure \ref{cp_dgp1_100}.}
\label{width_dgp2}
\end{center}
\end{figure}

\begin{remark} \label{rem:lag_augm}
We also considered the lag-augmentation approach proposed by \citet{KilianLuetkepohl17} and \citet{InoueKilian19} in the simulations where we implemented a variety of bootstrap versions in combination with lag-augmenting ($p=2$) the VAR; see Appendix \ref{sec:LAVAR} for details. 
Although the performance varied considerably with the specific bootstrap algorithm implemented, even the best performing lag-augmentation method did not seem able to account for rank uncertainty in a satisfactory way, and performed no better than the standard VAR in levels.\footnote{We also investigated the bias correction proposed by \citet{Kilian98b}.  We find that this method provides intervals with coverage close to nominal and comparable to the WIMP. However, some of the intervals are much wider than WIMP and even OLS intervals. The results are summarized in Appendix \ref{sec:LAVAR}.}
\end{remark}

\section{Fiscal Policy Shocks and Rank Uncertainty} \label{sec:fis}

We now study the potential ramifications of rank uncertainty on applied macroeconomic analysis. With our proposed approaches to construct inference accounting for rank uncertainty, we aim to assess the robustness of results obtained from unrestricted VARs. While there are countless VAR-based studies that use impulse response analysis to investigate the propagation of structural economic shocks, we focus in the following on fiscal policy shocks.

As our focus is methodological, we do not complement the literature on identification of structural VARs. Therefore we dispense with a detailed literature review on VAR-based policy analysis and only focus on evaluating seminal papers, reflecting various ways of identification. We also skip a detailed discussion of different identification approaches and their respective merits.\footnote{For a detailed exposition we refer to \citet{Ramey16} for a recent survey on various identification approaches and results in the literature.} Moreover, we omit any discussion on point estimates and focus solely on inference. \footnote{Our aim is not to challenge (widely accepted) empirical findings on the effects of economic policies, but to provide the applied researcher with tools that might help to construct more reliable inference. For that reason, we refrain from a simple replication exercise comparing different inferential approaches, and we want to stress that our goal is certainly not to contrast our findings to the original papers. Instead, we use the same reduced-form VAR and the same dataset across all applications, in order to move away from the original papers and only contrast results based on different identification procedures.}

Fiscal policy can relate to both the expenditure and revenue side of the government's budget. Measuring the effect of active spending policies as well as the consequences of tax changes has been an active field of economic research since decades. One of the first influential contributions using VAR-based impulse responses to assess the effect of government purchases is \citet{BlanchardPerotti02}. The authors identify spending shocks by a recursive identification scheme. With government spending ordered first, this translates into the assumption that government purchases are predetermined within the quarter.

Due to their assumed independence from general macroeconomic conditions, \citet{RameyShapiro98} construct narrative records  based on military buildups to identify truly exogenous spending changes. Those narrative time series have been embedded in several VAR studies and used to identify spending shocks by ordering this series first in a Cholesky-identified VAR. Among the most prominent studies following this approach is \citet{Ramey11}. In her paper she revisits the construction of the government spending news variable, filtering out possible distortions due to anticipation effects.

Narrative series have also been used to identify tax changes. In a series of papers \citet{MertensRavn11, MertensRavn12, MertensRavn13, MertensRavn14} construct various ``dis-aggregates'' of the \citet{RomerRomer09} measures of legislated changes in federal tax liabilities. Specifically, Mertens and Ravn distinguish between announced and unannounced tax changes, or between personal and corporate taxes. Moreover, they do not view those narrative series as a direct measure of ``tax-shocks'' but rather as an external \textit{proxy} which is correlated with the unknown structural shocks.\footnote{See also \citet{StockWatson12a} and \citet{MSW16}.}  Thus, instead of including the narrative variable in the VAR, one can obtain the structural shock of interest by regressing the narrative \textit{proxy} on the reduced-form residuals.

Yet another structural VAR identification approach imposes signs on the impulse responses to a particular shock for a certain horizon. \citet{MountfordUhlig09} identify a contractionary tax-shock as a shock, which leads to non-negative responses in government revenue during the first year after impact. Additionally, this tax-shock is identified by requiring it to be orthogonal to a business cycle shock and a monetary policy shock -- both identified through signs.\footnote{All shocks are identified sequentially by maximizing a penalty function which rewards responses in the desired direction and penalizes the others. Business cycle shocks are identified by assuming co-movements in the same direction as output, consumption, investment and government revenue. Contractionary monetary policy shocks affect responses in reserves and prices negatively and interest rate positively.} In particular, the orthogonality to business cycle fluctuations aims at controlling for movements in the government's budget caused by automatic stabilizers.

We compare uncertainty associated with the estimated impulse responses resulting from the above mentioned four identification approaches using the same data, and the same specification (as far as possible) of the underlying (reduced-form) VAR. That is, we use \citeauthor{BlanchardPerotti02}'s \citeyearpar{BlanchardPerotti02} structural VAR approach as well as \citeauthor{Ramey11}'s \citeyearpar{Ramey11} strategy to incorporate her narrative series in a VAR to identify the effect of government spending. Further, we use \citeauthor{MountfordUhlig09}'s \citeyearpar{MountfordUhlig09} sign-restriction scheme and \citeauthor{MertensRavn14}'s \citeyearpar{MertensRavn14} proxy-VAR to assess the effect of tax-shocks.

The choice of variables and the sample period is largely determined by the ``highest minimal requirement'' across the above identification approaches. The benchmark VAR is estimated in GDP,  private consumption, non-residential investment, government spending, (federal) tax receipts, total non-borrowed reserves, the federal funds rate, real wages, a price index, and the GDP deflator, where all variables except the federal funds rate are transformed to logs. The data is sampled quarterly from 1950/Q1 to 2006/Q4. A detailed description of the data is given in Appendix \ref{sec:data}. Additionally we use \citeauthor{Ramey11}'s \citeyearpar{Ramey11} news variable and \citeauthor{MertensRavn14}'s unanticipated tax-change proxy. The VAR representation in levels includes an intercept and a deterministic linear time trend, and four lags are included. We construct inference using the residual-based bootstrap algorithm presented in Algorithm \ref{alg:b}.\footnote{We did not find strong evidence of heteroskedasticity in the reduced-form residuals and refrain from using a robust bootstrap procedure such as the moving block bootstrap \citep{BJT16}. All approaches outlined in this paper could be easily extended in this way.}$^{,}$\footnote{While \citeauthor{Ramey11}'s \citeyearpar{Ramey11} news series is included in the VAR, and thus, bootstrapped ``endogenously'', we jointly draw (with replacement) from the reduced-form residuals and \citeauthor{MertensRavn14}'s external variable to account for uncertainty in estimating the effects of tax-shocks using this proxy.}

In order to make results somewhat comparable, impulse responses are normalized such that the point estimate of the response of the policy instruments has a peak at unity across different identification approaches \citep[see for example][]{Ramey11}. As a measure of uncertainty we plot 68\% confidence intervals, which is standard in the fiscal policy literature.\footnote{The data set as well as a MATLAB toolbox for the WIMP method with the identification schemes used used in this section are available at \href{http://www.stephansmeekes.nl}{http://www.stephansmeekes.nl}.}

\begin{sidewaysfigure}
\begin{minipage}{0.45\textwidth}
\begin{center}
\includegraphics[width=\linewidth]{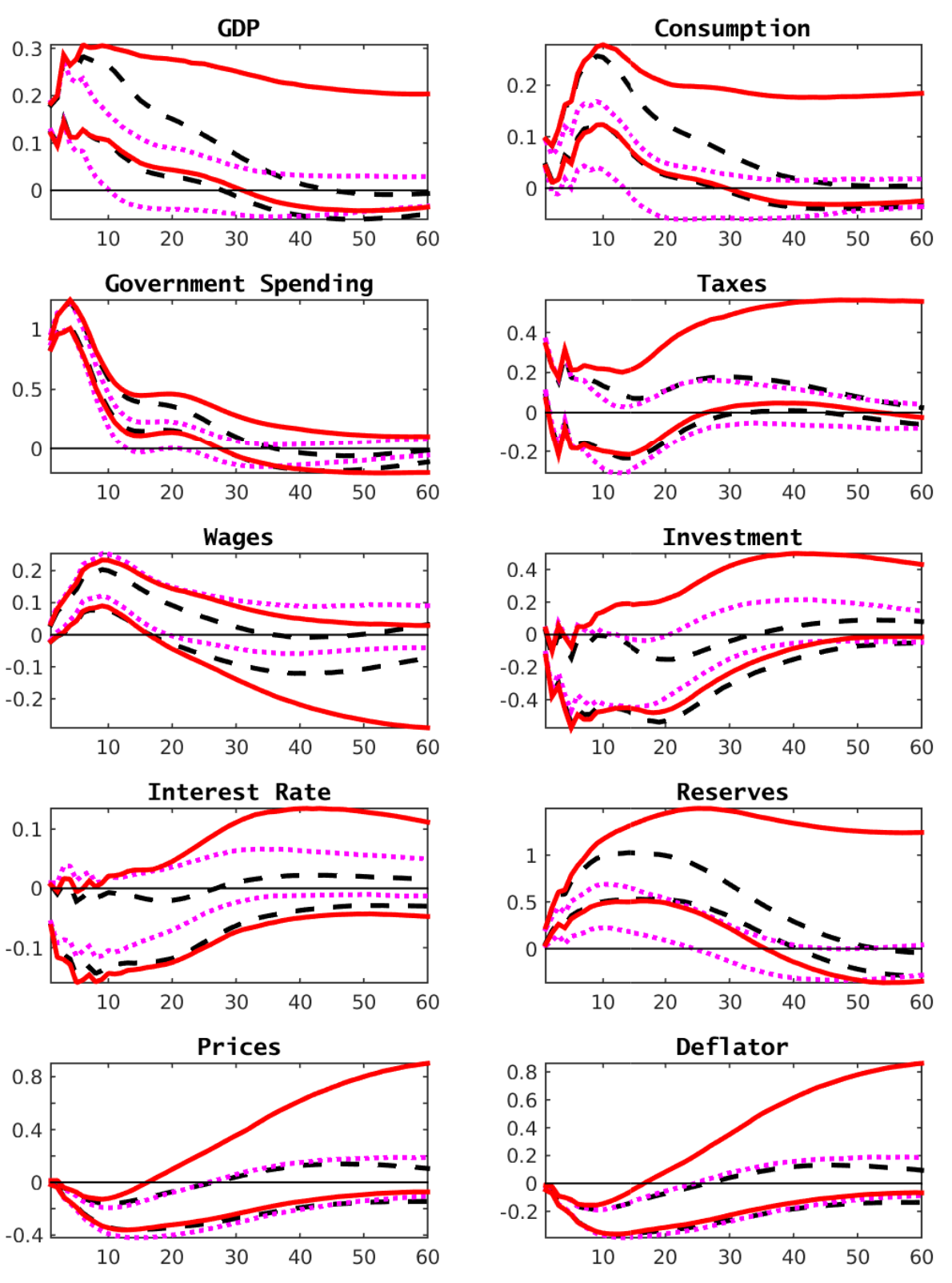}
\caption{68\% confidence intervals of impulse responses to a government spending shock identified as in \citet{BlanchardPerotti02}.
\textbf{Dashed} lines are OLS intervals, {\bf\color{HotPink2}{dotted}} lines FDBb/AIC intervals, {\bf\color{red}{solid}} lines  WIMP intervals.}
\label{Blanchard}
\end{center}
\end{minipage}
\hfill
\begin{minipage}{0.45\textwidth}
\begin{center}
\includegraphics[width=\linewidth]{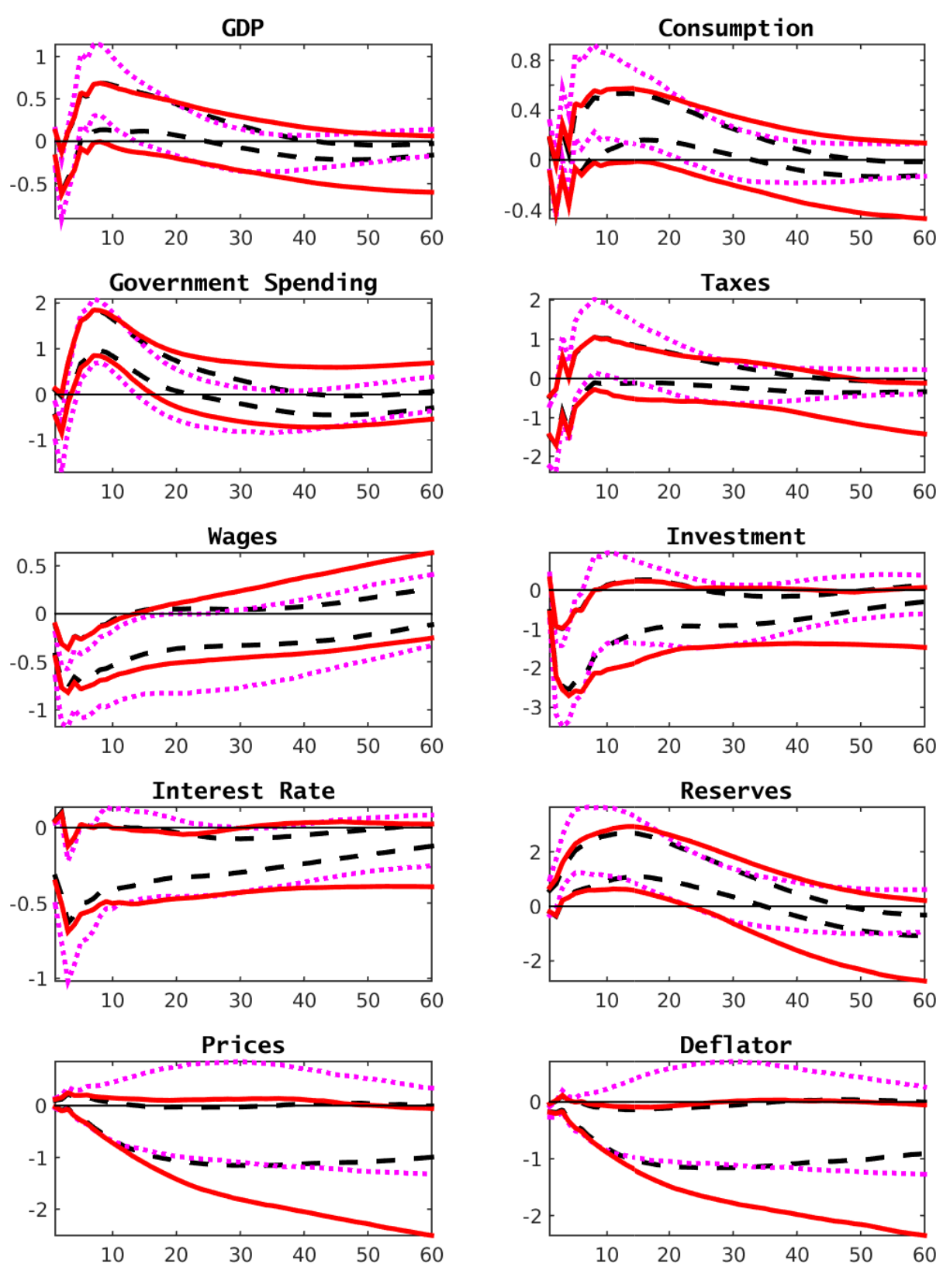}
\caption{68\% confidence intervals of impulse responses to a government spending shock identified as in \citet{Ramey11}.
For details see Figure \ref{Blanchard}.}
\label{Ramey}
\end{center}
\end{minipage}
\end{sidewaysfigure}

Figure \ref{Blanchard} and Figure \ref{Ramey} display unrestricted VAR in levels (estimated by OLS), FDB bagging (with AIC selection), and WIMP confidence bands (using the same specifications as in Section \ref{sec:sim}) of impulse responses due to a government spending shock. For the recursive VAR as in \citet{BlanchardPerotti02}, all three measures of uncertainty suggest that government spending shocks generate an initial boost in GDP. While the FDBb intervals indicate a rather moderate increase relative to the OLS intervals, the WIMP intervals imply maximum multiplier effects greater in range (roughly between 0.7 and 1.5). Considering impulse responses following Ramey's news shocks, it seems to be less clear whether government spending stimulates output or not. While the OLS confidence bands (and to a lesser extend the FDBb bands) support findings in the literature suggesting a short-lived boost in GDP, the WIMP intervals indicate greater uncertainty associated with the output response. Indeed, ``robust'' spending peak multipliers range between 0 and 3.3, such that a reliable conclusion on the effectiveness of spending policies cannot be made in this case.

\begin{sidewaysfigure}
\begin{minipage}{0.45\textwidth}
\begin{center}
\includegraphics[width=\linewidth]{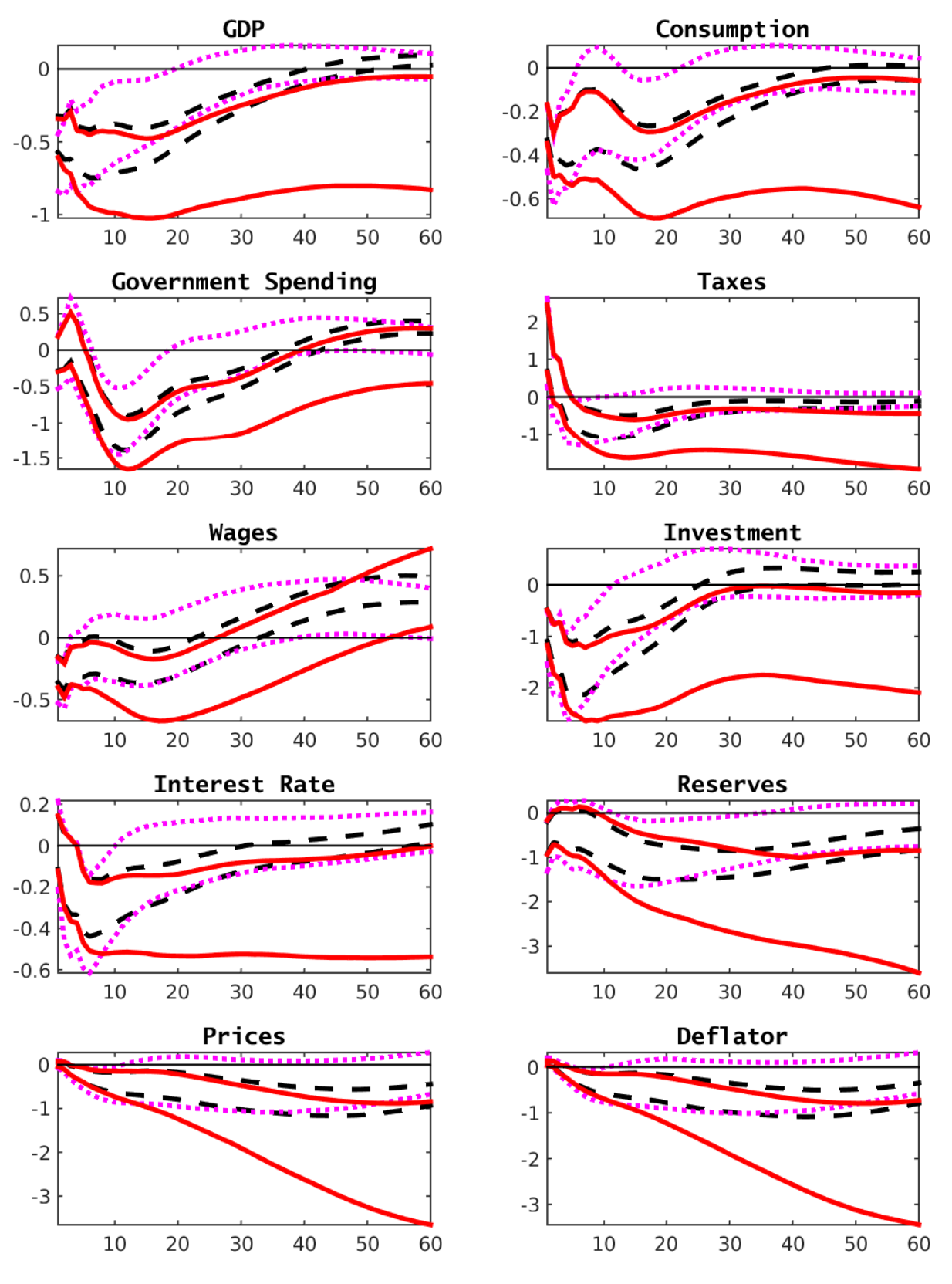}
\caption{68\% confidence intervals of impulse responses to a tax-shock identified as in \citet{MountfordUhlig09}.
For details see Figure \ref{Blanchard}.}
\label{Uhlig}
\end{center}
\end{minipage}
\hfill
\begin{minipage}{0.45\textwidth}
\begin{center}
\includegraphics[width=\linewidth]{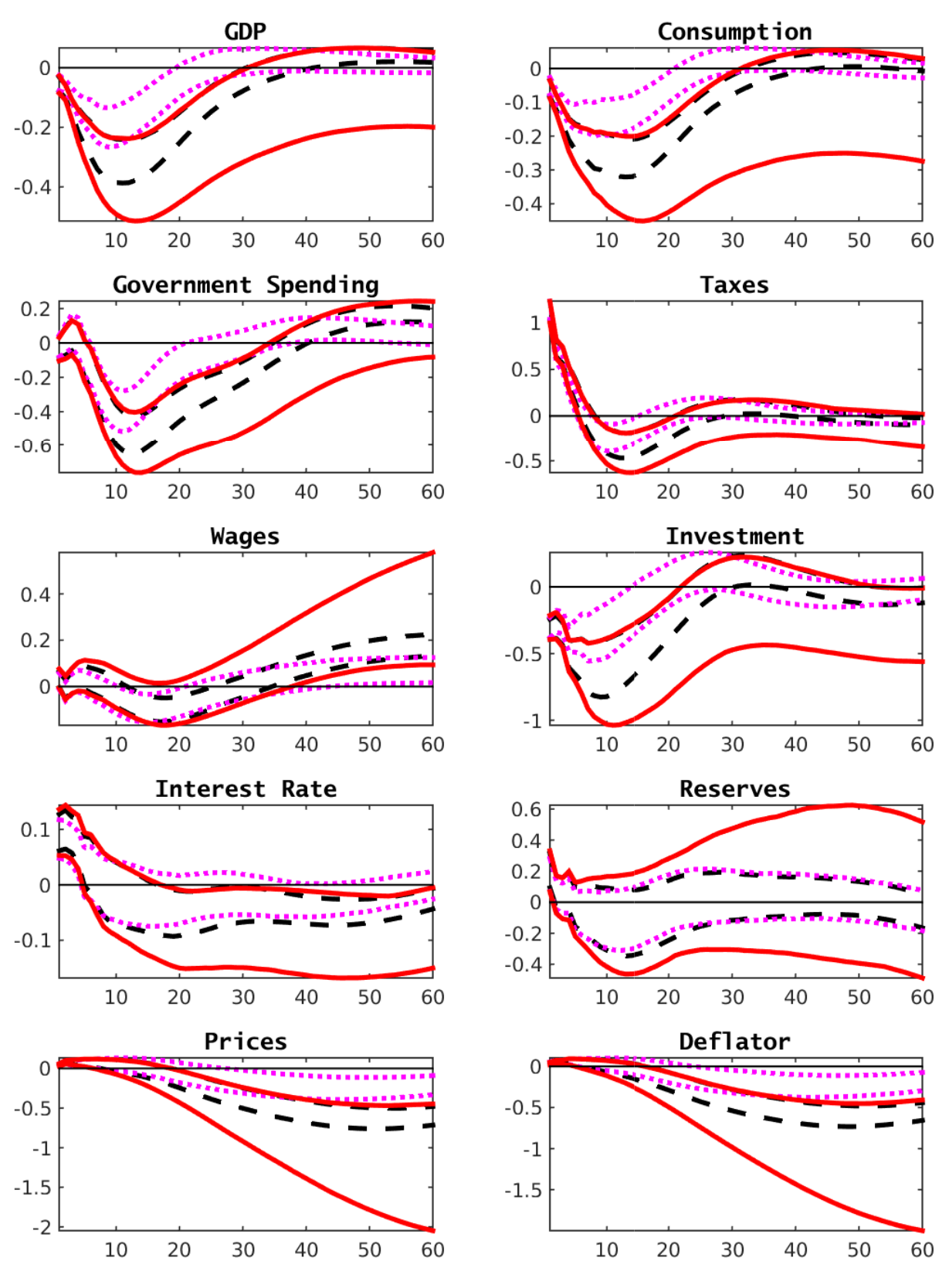}
\caption{68\% confidence intervals of impulse responses to a tax-shock identified as in \citet{MertensRavn12,MertensRavn14}.
For details see Figure \ref{Blanchard}.}
\label{Mertens}
\end{center}
\end{minipage}
\end{sidewaysfigure}

Confidence intervals of impulse responses following a contractionary tax-shock are displayed in Figures \ref{Uhlig} and \ref{Mertens}. Qualitatively, responses of GDP and its main aggregates are rather similar across both identification approaches and across all three inferential procedures: Output, consumption, and investment decrease significantly. The long-lived contraction in economic activity is accompanied by an equally lengthy decline in government spending, which hinders interpretation of the shocks as ``pure'' tax-shocks. Quantitatively, the implied response of output is much greater in the proxy VAR framework compared to the SVAR one. Intervals for peak multipliers include -6 for the former, and -3 for the latter.

Similar to the responses due to a government spending shock, the FDBb intervals are not necessarily wider than the OLS intervals. However, when considering the impact on output, and in contrast to scenario investigated above, the two intervals do not intersect at times and the FDBb intervals imply a significantly smaller impact on economic activity. This holds for both the shocks of \citet{MountfordUhlig09} and \citet{MertensRavn12,MertensRavn14}. Reflecting potentially more conservative inference, the WIMP intervals are wider, often encompassing the OLS intervals. Yet the WIMP intervals indicate that OLS-based inference rather underrates the effect of the identified tax-shocks on almost all variables. Generally, tax-shocks estimated by the proxy VAR imply greater effects on economic activity than those identified through sign restrictions. Moreover, the comparison with the spending shocks, supports some results in the literature suggesting that tax-cuts may be more effective in stimulating the economy.

\begin{figure}[!ht]
\begin{center}
\includegraphics[width=0.5\linewidth]{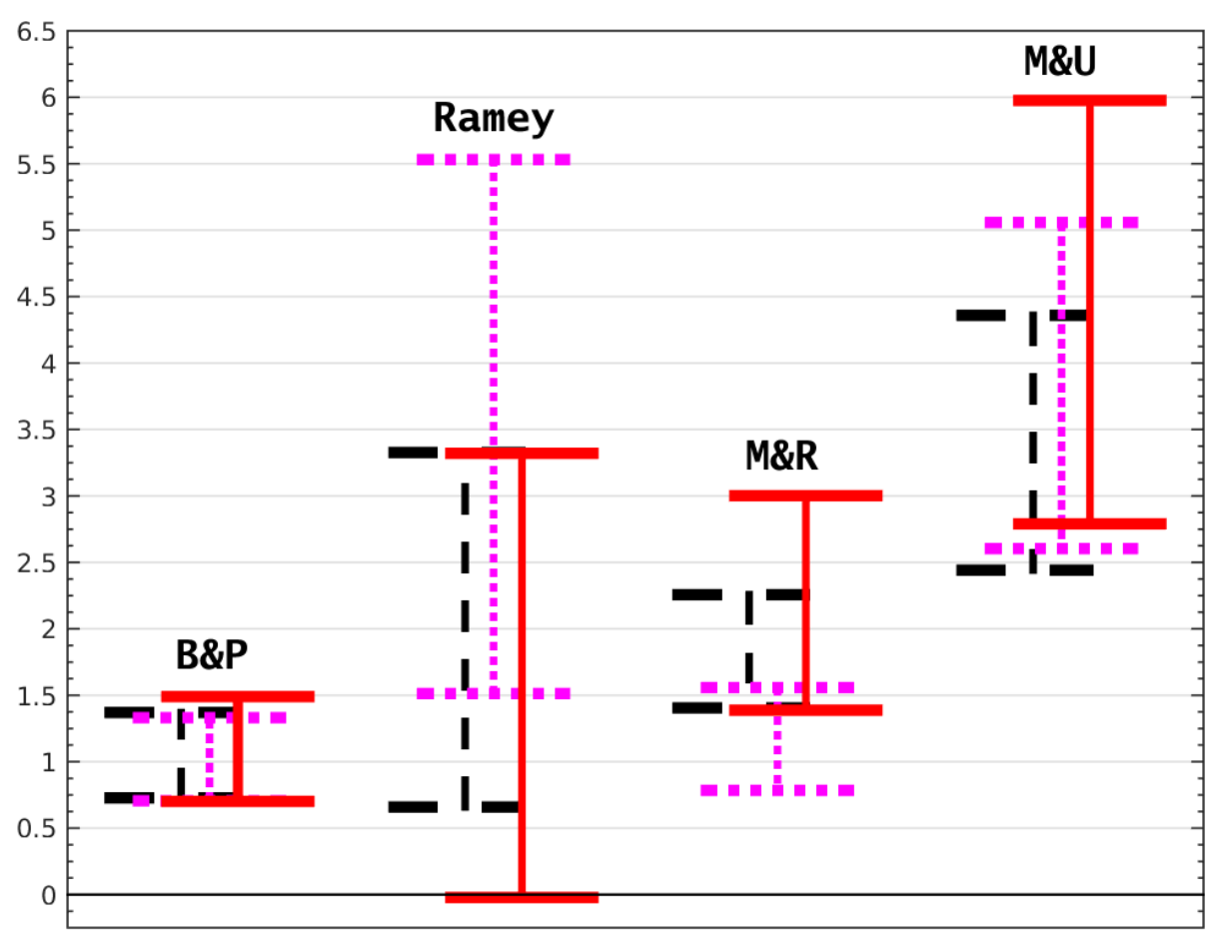}
\caption{68\% confidence intervals of peak multipliers implied by government spending and tax-cut shocks based on \citet{BlanchardPerotti02} \textbf{[B\&P]}, \citet{Ramey11} \textbf{[Ramey]}, \citet{MountfordUhlig09} \textbf{[M\&U]} and \citet{MertensRavn12,MertensRavn14} \textbf{[M\&R]}.
\textbf{Dashed} lines are OLS intervals, {\bf\color{HotPink2}{dotted}} lines FDBb/AIC intervals, {\bf\color{red}{solid}} lines WIMP intervals.}
\label{Multiplier}
\end{center}
\end{figure}

Figure \ref{Multiplier} compares confidence intervals for peak multipliers. Indeed, evidence suggesting that multipliers exceeding unity is much stronger for tax-cut policies than for spending policies. Based on the results for Ramey's news shock, multipliers due to expansionary spending policies might even not be significant at all.

The above results illustrate that ignoring uncertainty about the co-integration relations may lead to ambiguous quantification of statistical significance. Incorporating this uncertainty via the WIMP approach allows for a more confident interpretation of the results.

\section{Discussion} \label{sec:conc}
In this paper we have shown empirically and through a simulation study that ignoring uncertainty about cointegration relations may lead to unreliable inference for (structural) impulse responses.  Since the commonly used specification of the VAR in levels ignores any evidence for cointegration in the data, associated inference captures uncertainty only poorly. Also, model selection techniques, such as rank pre-estimation by sequential testing or information criteria, seem to deliver reliable inference only if evidence for the true cointegration rank is strong. In this paper we propose a novel data-driven approach to robust inference for impulse responses in the presence of uncertainty regarding the cointegration rank. Our WIMP approach is shown both by simulation and empirically to still be able to deliver meaningful (i.e.~not too wide) confidence intervals while being robust to rank uncertainty. As such it provides a reliable and simple alternative to the unreliable standard approaches.

Practical implementation of the WIMP approach only requires fixed-rank (bootstrap) intervals plus the sequence of trace tests for all rank tests, which are both readily available in any standard statistical software. While a toolbox for the WIMP methods used in our application is directly available, our approach can also easily be implemented for any desired SVAR analysis, as the fixed-rank intervals used as input for the WIMP can be based on any appropriate method, both in terms of inference method such as the bootstrap and identification scheme. Finally, the computational cost of the method is fairly low; on any modern computer bootstrap intervals for a fixed rank are fast to compute, and given that in this kind of VAR model the number of variables (and hence the number of ranks) has to be relatively low to avoid the curse of dimensionality, doing so for all ranks should pose no problem.

While prudent construction of inference is particularly important for impulse responses, our proposed WIMP procedure is equally beneficial in different VAR contexts, such as forecasting. While forecast combinations across different models are well accepted as point forecasts, our WIMP method allows to construct corresponding interval forecasts that account for model uncertainty. More generally, the approach can be adapted to a variety of model selection problems, as long as the relative evidence for a particular model can be assessed against a modest number of alternatives. While in theory it can be applied to high-dimensional problems as well, computationally the method is best suited for low-dimensional problems where the number of models is relatively small. While this is a limitation of the method, it is inherent to the simultaneous inference philosophy behind, which also holds for the PoSI method of \citet{PoSI13}. Exploring the usefulness and limitations of the WIMP in more general settings is therefore an interesting avenue for future research.

\singlespacing

\newpage
\onehalfspacing
\begin{appendices}

\numberwithin{algorithm}{section}
\numberwithin{figure}{section}
\numberwithin{remark}{section}
\section{Algorithms} \label{sec:alg}
Here we describe the bootstrap algorithms used in the paper. Algorithm \ref{alg:b} is the specific fixed-rank bootstrap algorithm used in the simulation and empirical sections.

\begin{algorithm} [Bootstrap Confidence Interval under Rank $r$] \label{alg:b}
\leavevmode \vspace{-\baselineskip}	
\begin{enumerate}
\item Let $\tilde{y}_t = y_t - \hat{\mu}_0 - \hat{\mu}_1 t$ for $t=1,\ldots,T$ and estimate the VECM under rank $r$ and obtain the residuals
\begin{equation*}
\hat{u}_t = \Delta \tilde{y}_t - \hat{\Pi}^{(r)} \tilde{y}_{t-1} - \sum_{j=1}^{p-1} \hat{\Gamma}_j^{(r)} \Delta \tilde{y}_{t-j}, \qquad t=p+2, \ldots,T.
\end{equation*}

\item Use a bootstrap method to obtain bootstrap errors $\left\{u_t^{*}\right\}_{t=p+2}^{T}$ from the residuals $\left\{\hat{u}_t\right\}_{t=p+2}^T$.

\item Build the bootstrap sample $\left\{y_t^{*}\right\}_{t=1}^T$ recursively as
\begin{equation*}
y_t^{*} = y_{t-1}^{*} + \hat{\Pi}^{(r)} y_{t-1}^{*} + \sum_{j=1}^{p-1} \hat{\Gamma}_j^{(r)} \Delta y_{t-j}^{*} + u_t^{*}, \qquad t = p+2,\ldots,T,
\end{equation*}
using initial values $y_1^*, \ldots, y_{p+1}^*$.

\item Detrend the bootstrap sample to obtain $\tilde{y}_t^* = y_t^* - \hat{\mu}_0^* - \hat{\mu}_1^* t$ for $t=1,\ldots,T$. Estimate the VECM under rank $r$ on $\{\tilde{y}_t^{*}\}_{t=1}^T$ to obtain $\hat{\theta}^{(r)*}$. Obtain the bootstrap impulse response as $\hat{\zeta}^{(r)*} = \bar{f}(\hat{\theta}^{(r)*})$.

\item Repeat Steps 2 to 4 $B$ times. Let $q^* (\gamma)$ denote the $\gamma$-quantile of the $B$ centered bootstrap statistics $\hat{\zeta}^{(r)*} - \hat{\zeta}^{(r)}$. Construct a $(1-\gamma)$-confidence interval for $\zeta$ as $\left[L^{(r)} (\gamma), U^{(r)} (\gamma) \right]$, where $L^{(r)} (\gamma) = \hat{\zeta}^{(r)} - q^{*} (1-\gamma/2)$ and $U^{(r)} (\gamma) = \hat{\zeta}^{(r)} - q^{*} (\gamma/2)$.
\end{enumerate}
\end{algorithm}

\begin{remark}
Depending on the specific assumptions made on $\{u_t\}$, a variety of different bootstrap methods, such as i.i.d., wild or block bootstrap, can be used in Step 2 of Algorithm \ref{alg:b}; we provide further details in Section \ref{sec:wimp_av}. Similarly, different initializations in Step 3 can be used. For the simulation study and application in this paper, we use the i.i.d.~bootstrap in Step 2 and initialize the bootstrap sample in step 3 by setting $y_t^* = y_t$ for $t=1,\ldots,p+1$.
\end{remark}

\begin{remark} \label{rem:detr}
Instead of detrending or demeaning (with $\hat{\mu}_1=0$) prior to estimation, one could also directly incorporate deterministic components in the VECM \citep[cf.][]{Johansen95}. However, one then has to decide how the deterministic components affect the long run and short run components separately, resulting in a multitude of different specifications. Our simpler, robust, strategy corresponds to the typical approach taken in most empirical studies, and makes the estimators of the detrended VECM invariant to the true deterministics present in the DGP.

In Step 4 of the algorithm we detrend the bootstrap data again, re-estimating the deterministic components, which might appear unnecessary as the bootstrap data do not contain any trends. However, this is done to mimic the effect of detrending on the calculated impulse responses, which under cointegration and at very long horizons, will affect the asymptotic distributions as it would unit root or cointegration analyses. It might be tempting to also first ``retrend'' the bootstrap data, that is, to put the estimated trend back into the bootstrap sample. This is however unnecessary as the consequent detrending makes the estimators invariant to the exact value of the trend coefficient, see for example Remark 2 in \citet{Smeekes13}.
\end{remark}

\noindent Algorithm \ref{alg:b_ers} shows how endogenous rank selection can be implemented in the bootstrap.

\begin{algorithm} [Bootstrap Endogenous Rank Selection (BERS)] \label{alg:b_ers}
Choose a rank selection method $M_r(\cdot)$, and let $\hat{r} = M_r (Y_T)$. Perform Steps 1-3 of Algorithm \ref{alg:b} with $r = \hat{r}$ or $r = K$. Next, replace Step 4 by
\begin{enumerate}
\setcounter{enumi}{3}
\item Let $\hat{r}^* = M_r (Y_T^*)$, where $Y_T^* = (y_1^*, \ldots, y_T^*)^\prime$. Estimate the VECM with rank $\hat{r}^*$ on the bootstrap sample $(y_t^{*})_{t=1}^T$ (after detrending) to obtain $\hat{\theta}^{(\hat{r}^*)*}$. Obtain the bootstrap impulse response as $\hat{\zeta}^{(\hat{r}^*)*} = \bar{f} (\hat{\theta}_j^{(\hat{r}^*)*})$.
\end{enumerate}
Perform Step 5 as in Algorithm \ref{alg:b}.
\end{algorithm}

\noindent Algorithm \ref{alg:b_fdb} details how to implement bagging with the Fast Double Bootstrap.
\begin{algorithm} [FDB bagging (FDBb)] \label{alg:b_fdb}
Choose a rank selection method $M_r (\cdot)$, and perform steps 1-4 of Algorithm \ref{alg:b_ers}. Next:
\begin{enumerate}
\setcounter{enumi}{4}
\item Perform a second bootstrap procedure on the bootstrap sample $\{y_t^{*}\}_{t=1}^T$ to obtain double-bootstrap impulse responses. For every bootstrap sample $\{y_t^{*}\}_{t=1}^T$, only \textit{one} second-level bootstrap sample has to be drawn. Specifically, take the following steps:

\begin{enumerate}[{\bf (i)}]
\item Estimate the VECM with rank $\hat{r}^* = M_r(Y_T^*)$, where $Y_T^* = (y_1^*, \ldots, y_T^*)^\prime$, and obtain the residuals
\begin{equation*}
\hat{u}_t^{*} = \Delta \tilde{y}_t^{*} - \hat{\Pi}^{(\hat{r}^*)*} \tilde{y}_{t-1}^* - \sum_{j=1}^p \hat{\Gamma}_j^{(\hat{r}^*)*} \Delta \tilde{y}_{t-j}^{*}, \qquad t= p+2, \ldots, T.
\end{equation*}

\item Construct the second-level bootstrap errors $\left\{u_t^{**}\right\}_{t=p+2}^{T}$ from $\left\{\hat{u}_t^{*}\right\}_{t=p+2}^T$ using the same bootstrap method as for the first level, and build the second-level bootstrap sample $\left\{y_t^{**}\right\}_{t=1}^T$ recursively as
\begin{equation*}
y_t^{**} = y_{t-1}^{**} + \hat{\Pi}^{(\hat{r}^*)*} y_{t-1}^{**} + \sum_{j=1}^p \hat{\Gamma}_j^{(\hat{r}^*)*} \Delta y_{t-j}^{**} + u_t^{**}, \qquad t = p+2, \ldots,T.
\end{equation*}

\item Estimate the cointegration rank $\hat{r}^{**} = M_r(Y_T^{**})$, where $Y_T^{**} = (y_1^{**}, \ldots, y_T^{**})^\prime$. Estimate a VECM with rank $\hat{r}^{**}$ on $Y_T^{**}$ to obtain the double-bootstrap impulse responses $\hat{\zeta}^{(\hat{r}^{**})**}$.
\end{enumerate}

\item Repeat Steps 1 to 5 $B$ times. Let $\hat{\zeta}_{1}^{(\hat{r}^*)*}, \ldots, \hat{\zeta}_{B}^{(\hat{r}^*)*}$ denote the ordered sequence of the first-level bootstrap estimates obtained over the $B$ bootstrap replications. The \textit{bagging} estimator of the impulse response is then defined as $\hat{\zeta}^{\text{bag}} = B^{-1} \sum_{b=1}^{B} \hat{\zeta}_{b}^{(\hat{r}^*)*}$. Let $q^{**} (\gamma)$ denote the $\gamma$-quantile of the $B$ centered second-level bootstrap statistics $\hat{\zeta}^{(\hat{r}^{**})**} - \hat{\zeta}^{(\hat{r}^{*})*}$. Construct a $(1-\gamma)$-confidence interval for $\zeta$ as $\left[\hat{\zeta}^{\text{bag}} - q^{**} (1-\gamma/2), \hat{\zeta}^{\text{bag}} - q^{**} (\gamma/2) \right]$.
\end{enumerate}
\end{algorithm}

\section{Proofs}

\begin{proof}[Proof of Theorem \ref{th:av_pw}]
By Assumption (i), we have that $\Prob (R = r_0) \rightarrow 1$. As by construction $L^{\WIMP} (\gamma) \leq L^{(R)}$, it follows that
\begin{equation*}
\begin{split}
\Prob \left(L^{\WIMP} (\gamma) \leq L^{(r_0)} (\gamma) \right) &= \Prob \left( \left.L^{\WIMP} (\gamma) \leq L^{(R) }(\gamma) \right| R = r_0 \right) \Prob (R = r_0) + o(1) \\
&= P(R = r_0) + o(1) \rightarrow 1,\\
\end{split}
\end{equation*}
and similarly $\Prob \left(U^{\WIMP} (\gamma) \geq U^{(r_0)} (\gamma)\right) \rightarrow 1$. The result then follows from assumption (ii) as
\begin{equation*}
\Prob \left( L^{\WIMP} (\gamma) \leq \zeta \leq U^{\WIMP} (\gamma) \right) \geq \Prob \left(L^{(r_0)} (\gamma) \leq \zeta \leq U^{(r_0)}(\gamma) \right) +o(1) \rightarrow 1 - \gamma.\qedhere
\end{equation*}
\end{proof}

\begin{proof}[Proof of Proposition \ref{lem:weights}]
It follows from \citet{Johansen95} and \citet{BernsteinNielsen14} that for all $r \geq r_0$,  $J_T (r) = O_p(1)$, such that $T^{-c_2} J_T(r) \xrightarrow{p} 0$, while for $r < r_0$, we have that $J_T (r)/T$ is tight, such that $T^{-c_2} J_i(r) = T^{1-c_2} J_T(r) / T \xrightarrow{p} \infty$. Therefore we have that $e^{-c_1 T^{-c_2} J_T(r)} \xrightarrow{p} \mathbbm{1}(r \geq r_0)$ and consequently $W_T(r) \xrightarrow{p} \mathbbm{1}(r = r_0)$.
\end{proof}

\section{Additional Simulations}\label{sec:LAVAR}
In this section we investigate by simulation the properties of two alternative bootstrap approaches, the lag-augmentation proposed by \citet{KilianLuetkepohl17} and \citet{InoueKilian19} as well as the bias correction of \citet{Kilian98}.

In order to combine the idea of lag-augmentation with a bootstrap algorithm several choices have to be made. In particular, the VAR process from which the bootstrap samples are built (see e.g.~Step 1/Step 3 of Algorithm \ref{alg:b}) has to be specified. Potential candidates are the ``correctly specified'' VAR and the lag-augmented one. Similarly, one has to decide how to re-estimate the VAR parameters from each bootstrap sample, i.e. using lag-augmentation or not. \citet{KilianLuetkepohl17} and \citet{InoueKilian19} provide little practical guidance on these decisions. We investigated various possible ways to generate bootstrap inference with lag-augmented VARs and found that the performance heavily varied with these choices. We here report the performance of the best performing method, where in Step 1 in Algorithm \ref{alg:b}, we estimate a (correctly lag-specified) VAR(1) in levels to construct the bootstrap DGP in Step 3, whereas $\hat{\zeta}$ and $\hat{\zeta}^*$ in Step 4 and 5 are based on (the first lag of) a lag-augmented VAR in levels ($p=2$) estimated on the data $Y_T$ and the simulated data $Y_T^*$, respectively.

Figure \ref{cp_dgp1_la} and Figure \ref{width_dgp1_la} summarize the simulation results for the lag-augmented approach, including OLS and WIMP as well for ease of comparison. The empirical coverage probabilities of the lag-augmented VAR in Figure \ref{cp_dgp1_la} are reasonably good, especially for longer horizons. For short horizons, the intervals have more serious undercoverage than OLS and WIMP. However, the widths of the intervals in Figure \ref{width_dgp1_la} show that the lag-augmented VAR intervals are much wider than OLS and WIMP intervals, and clearly far too wide to be of any practical use. This is because the lag-augmented VAR tends to imply overly persistent, often explosive dynamics -- at least in our simulation design.

\begin{figure}[!ht]
\begin{center}
\includegraphics[width=0.6\linewidth]{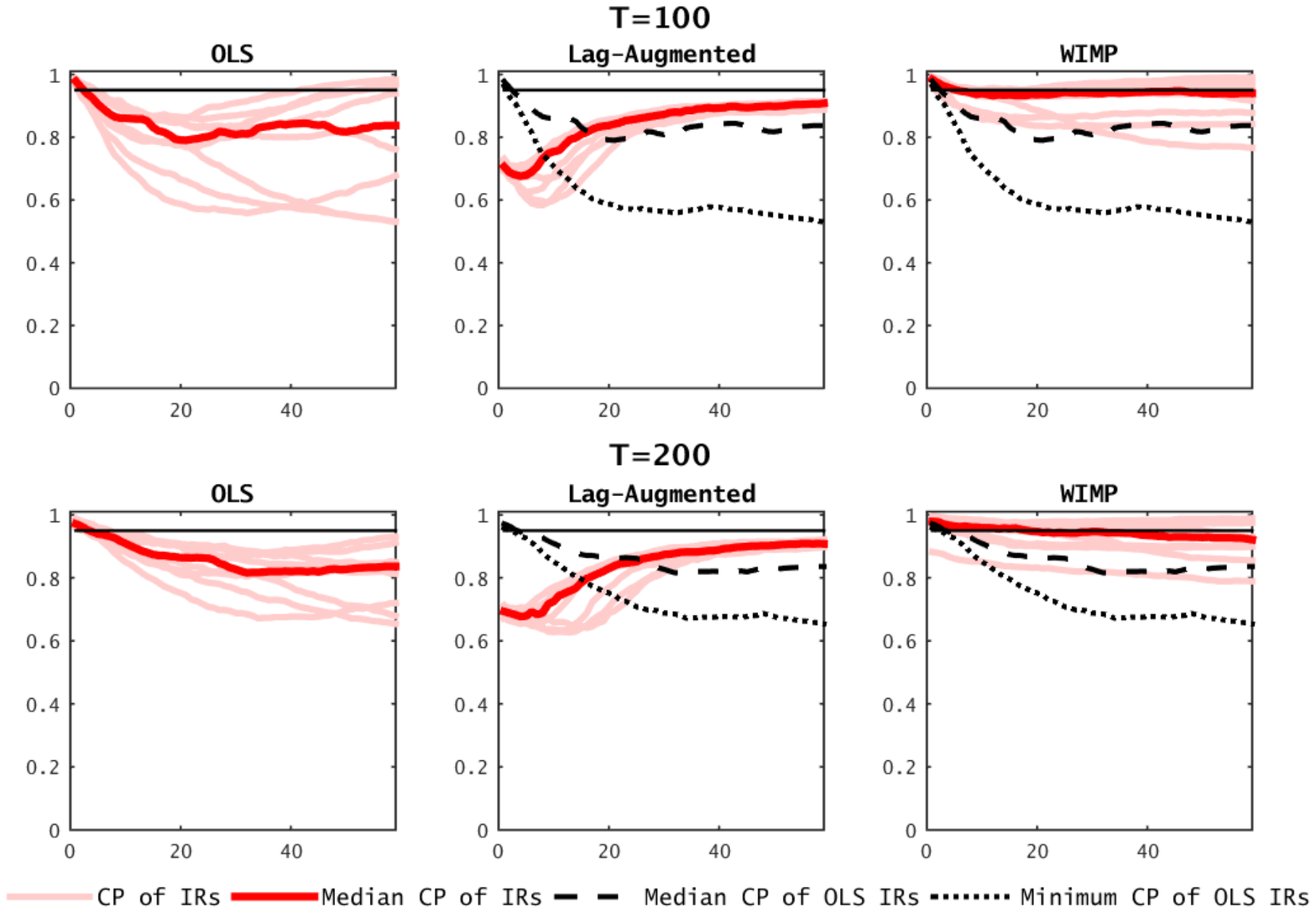}
\caption{DGP1: Empirical coverage rates for the lag-augmented approach for $T=100$ and $T=200$. For details see Figure \ref{cp_dgp1_100}.}
\label{cp_dgp1_la}
\end{center}
\begin{center}
\includegraphics[width=0.7\linewidth]{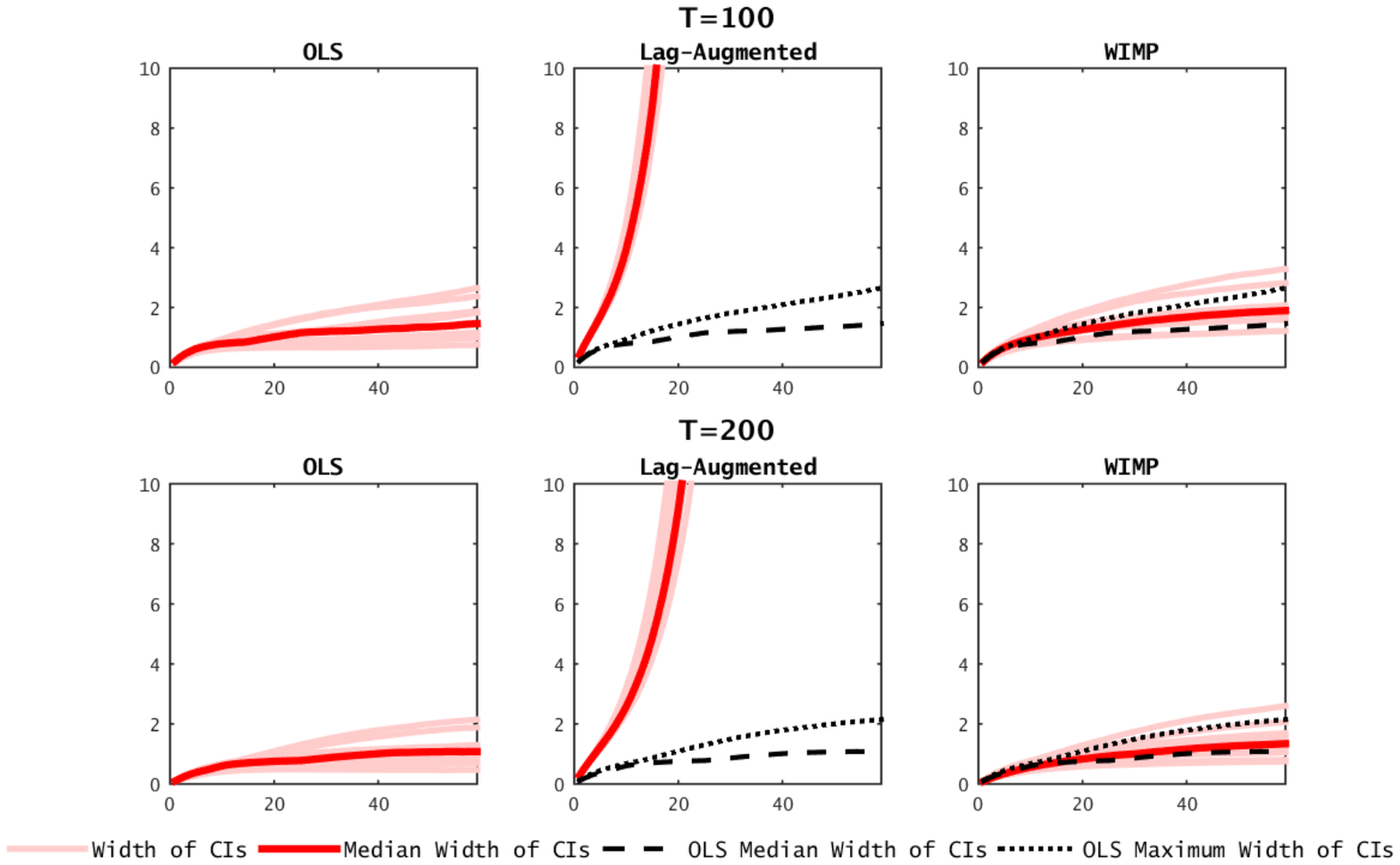}
\caption{DGP1: Average width of 95\% bootstrap CIs for the lag-augmented approach for $T=100$ and $T=200$. For details see Figure \ref{cp_dgp1_100}.}
\label{width_dgp1_la}
\end{center}
\end{figure}

\citet{InoueKilian19} suggest to use lag-augmentation in combination with a small sample bias-correction as proposed in \citet{Kilian98}. However, it is (again) not entirely clear how to best combine both procedures in the bootstrap algorithm. Thus, we investigate the implication of Kilian's bias-correction separately. We follow the algorithm in \citet{Kilian98} in combination with a i.i.d. bootstrap.

The empirical coverage probabilities of the bias-corrected VAR in Figure \ref{cp_dgp2_bc} are close to their nominal levels and comparable to those of the WIMP. As displayed in Figure \ref{width_dgp2_bc}, some of the bias-corrected VAR intervals are, however, much wider than WIMP and even OLS intervals. 
 
\begin{figure}[!ht]
\begin{center}
\includegraphics[width=0.6\linewidth]{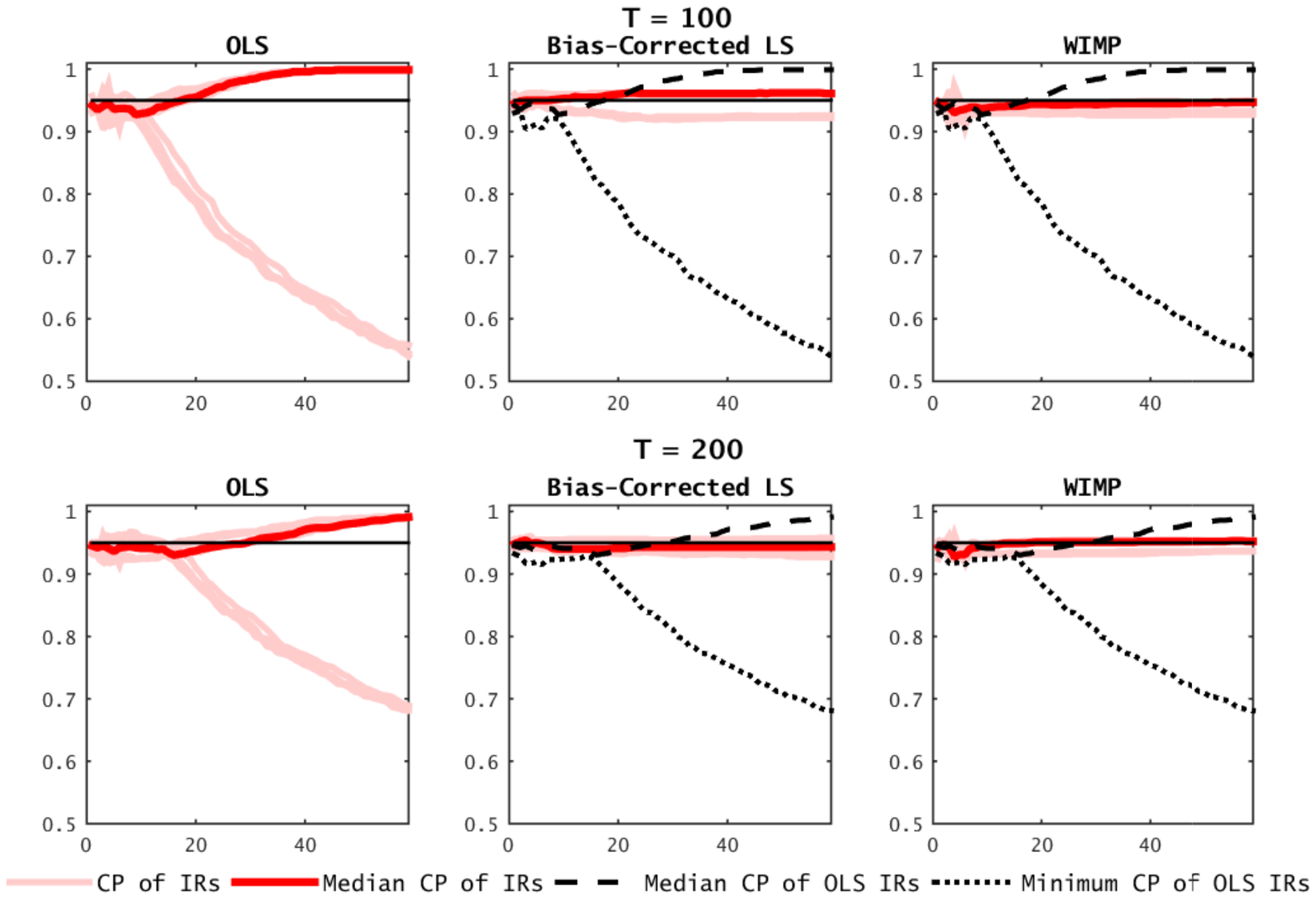}
\caption{DGP2: Empirical coverage rates for the bias-corrected approach for $T=100$ and $T=200$. For details see Figure \ref{cp_dgp1_100}.}
\label{cp_dgp2_bc}
\end{center} 
\begin{center}
\includegraphics[width=0.7\linewidth]{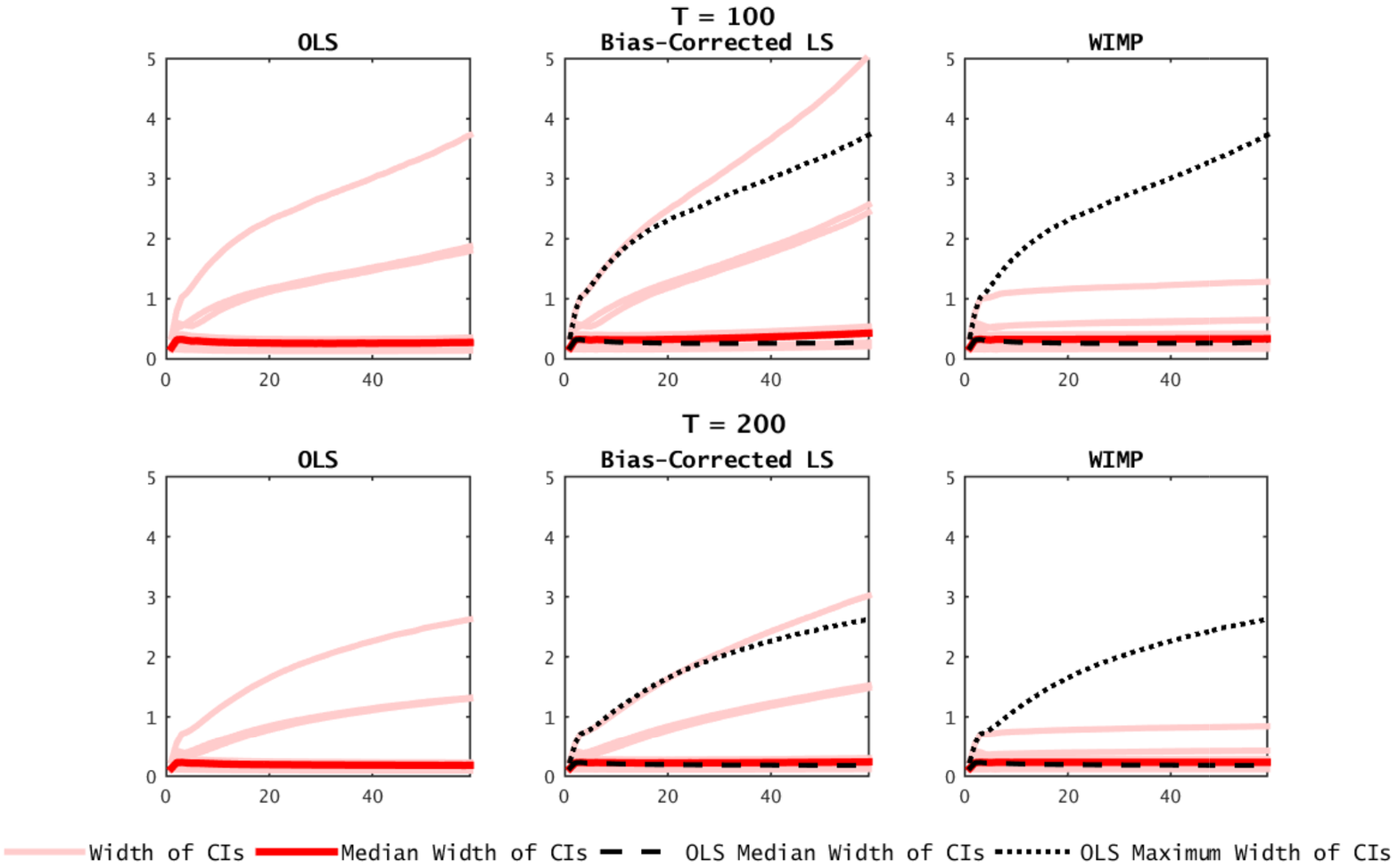}
\caption{DGP2: Average width of 95\% bootstrap CIs for the bias-corrected approach for $T=100$ and $T=200$. For details see Figure \ref{cp_dgp1_100}.}
\label{width_dgp2_bc}
\end{center}
\end{figure}

\newpage

\section{Data} \label{sec:data}
All data is quarterly, sampling from 1950/Q1-2006/Q4. We composed the data from three sources: The Bureau of Economic Analysis' \textit{U.S. National Income and Product Accounts} (NIPA) (\href{https://www.bea.gov/national}{bea.gov/national}), The Bureau of Labor Statistics (BLS) (\href{https://www.bls.gov}{bls.gov}), and \textit{FRED Economic Database} hosted by the Federal Reserve Bank of St.~Louis (\href{https://fred.stlouisfed.org}{fred.stlouisfed.org}).
\begin{description}
\item \textbf{GDP} is taken from NIPA table 1.1.5.
\item \textbf{Consumption} is \textit{private consumption}, NIPA table 1.1.5.
\item \textbf{Investment} is \textit{gross private non-residential investment}, NIPA table 1.1.5.
\item \textbf{Government spending} is \textit{government expenditure and gross investment}, NIPA table 3.9.5.
\item \textbf{Taxes} are \textit{Federal government current tax receipts} plus \textit{contributions for social insurance} minus \textit{income taxes from federal reserve banks}, all in NIPA table 3.2.
\item \textbf{Real wages} are \textit{nonfarm business sector: real compensation per hour}, from the BLS.
\item \textbf{GDP deflator} is taken from NIPA table 1.1.9.
\item \textbf{Federal funds rate} is taken from FRED, series code: \textit{fedfunds}.
\item \textbf{Adjusted reserves} is taken from FRED, series code: \textit{ADJRESSL}.
\end{description}
GDP and its components, government revenue, and adjusted reserves are transformed into real per capita values using the GDP deflator and a population measure (NIPA table 7.1).

\end{appendices}

\end{document}